\documentclass{article}

\usepackage{graphicx}
\usepackage{algorithm}
\usepackage{algpseudocode}
\usepackage{amsmath}
\usepackage{amssymb}
\usepackage{amsthm}
\usepackage{authblk}
\usepackage[labelformat = simple]{subfig}

\usepackage[disable]{todonotes}

\newtheorem{lemma}{Lemma}
\newtheorem{theorem}{Theorem}
\theoremstyle{definition}
\newtheorem{definition}{Definition}
\theoremstyle{remark}
\newtheorem{observation}{Observation}

\newcommand{\A}{\mathcal{A}}
\newcommand{\C}{\hat C}
\renewcommand{\P}{\hat P}
\newcommand{\R}{\mathbb R}

\newcommand{\al}{\alpha}

\newcommand{\lm}{\lambda}

\newcommand{\tb}[1]{\textbf{#1}}

\newcommand{\lct}{\operatorname{l-cost}}
\newcommand{\cl}{\operatorname{close}}
\newcommand{\ct}{\operatorname{cost}}
\newcommand{\cv}{\operatorname{cover}}

\newcommand{\nr}{\operatorname{near}}
\newcommand{\rank}{\operatorname{rank}}

\graphicspath{{"./images/"}}

\begin{document}

\title{A Subquadratic Time Algorithm for the Weighted $k$-Center Problem on Cactus Graphs}





\author[1]{Binay Bhattacharya}
\affil[1]{Simon Fraser University, Burnaby, Canada}
\author[2]{Sandip Das}
\author[2]{Subhadeep Ranjan Dev}
\affil[2]{Indian Statisitical Institute, Kolkata, India}

\maketitle


\begin{abstract}
    The weighted $k$-center problem in graphs is a classical facility location problem where we place $k$ centers on the graph, which minimize the maximum weighted distance of a vertex to its nearest center. We study this problem when the underlying graph is a cactus with $n$ vertices and present an $O(n \log^2 n)$ time algorithm for the same. This time complexity improves upon the $O(n^2)$ time algorithm by Ben-Moshe et al. \cite{ben}, which is the current state-of-the-art.

\end{abstract}

\section{Introduction} \label{sec:intro}

The facility location problem is a classical problem in operations research which has been well studied in both the graph and geometry settings. In graph, the vertices have some kind of demand, and the edges have some associated length. The length on the edges correspond to the distance between its two end vertices. The objective is to identify a certain number of locations on the edges and vertices of the graph where we can place facilities. The facilities are placed such that they satisfy the demands on the vertices of the graph and also minimize a certain predefined objective function.

In the \emph{weighted $k$-center problem} the objective is to place $k$ facilities which minimize the maximum weighted distance between a vertex and its nearest facility. Kariv and Hakimi \cite{kariv} introduced this problem in 1979 and showed it to be NP-hard for general graphs. In fact, they proved a much stronger statement where they showed that the problem is NP-hard on planar graphs of maximum degree 3. However, the problem is shown to have efficient polynomial runtime when either the graph class is restricted or there is some bound on $k$. For example, in trees, Wang and Zhang \cite{wang} recently showed that the weighted $k$-center problem can be solved in $O(n \log n)$ time, where $n$ is the number of vertices in the tree. If $k$ is a constant, Bhattacharya et al. \cite{bhattacharya2019} presented an optimal $O(n)$ time solution for the problem on trees. For cactus, the current state of the art is an $O(n^2)$ time algorithm by Ben-Moshe et al. \cite{ben}. When $k$ is restricted to be 1 or 2 then the runtime for cactus comes down to $O(n \log n)$ and $O(n \log^3 n)$ time respectively \cite{bhattacharya2017}. For partial $p$-trees, Bhattacharya and Shi \cite{bhattacharya} showed that the weighted $k$-center problem can be solved in $O(n^k \log^{2p} n)$ time, when the facilities are restricted to be placed only on the vertices.

The algorithm of Wang and Zhang \cite{wang} for the weighted $k$-center problem on trees uses Frederickson's parametric search technique \cite{frederickson1991optimal} to bring down the runtime from $O(n \log n + k \log^2 n \log(n/k))$ \cite{banik} to $O(n \log n)$. We generalize Frederickson's technique to present an $O(n \log^2 n)$ time algorithm for the problem. This is the first time that such a generalization could be achieved for a graph class other than trees and is the first subquadratic time algorithm for the problem. The feasibility test is the decision version of the problem where given a cost $\lm$, the test returns whether the optimal cost (not known) is at most $\lm$ or not. We have also proposed an algorithm for the feasibility test with respect to the weighted $k$-center problem, which runs in optimal $O(n)$ time.

The rest of the paper is organized as follows.  In Section \ref{sec:prelim}, we provide a formal definition of the weighted $k$-center problem and the tree representation of a cactus. We present some results on circular arc graphs which we use in the later sections. In Section \ref{sec:feasibility}, we present our algorithm for the feasibility test. For any input $\lambda$, the test determines whether $\lambda \geq \lambda^*$ or not, where $\lambda^*$ is the optimal cost of the weighted $k$-center problem. In Section \ref{sec:arrangements}, we prove some properties on linear arrangements whose results we use in Section \ref{sec:stem}. In Section \ref{sec:stem}, we define a cactus stem which is basically a subgraph of the cactus and a generalization of the (path) stem as defined in Frederickson \cite{frederickson1991optimal}. We then present an $O(n \log^2 n)$ time algorithm to solve the weighted $k$-center problem on the cactus stem. We also show another technique with which we can get an $O(n \log n + k n \log n)$ time algorithm. In Section \ref{sec:main}, we present our $O(n \log^2 n + k n \log n)$ time algorithm for cacti. Finally we conclude in Section \ref{sec:conclusion}.

\section{Preliminaries} \label{sec:prelim}
\subsection{Problem definition}
Let $G$ be a graph with vertex set $V(G)$ and edge set $E(G)$ with $|V(G)| = n$. We say $G$ is a \emph{cactus} if no two cycles of $G$ share a common edge. Each vertex $v \in V(G)$ is associated with a non-negative weight $w(v)$ and each edge $e \in E(G)$ is associated with a positive length $l(e)$. We interpret $e$ as a line segment of length $l(e)$ so that any point $x$ on $e$ can be referred to. The set of all points on all edges of $G$ is denoted by $A(G)$. For any two points $x, y \in A(G)$, the distance between $x$ and $y$, denoted by $d_G(x, y)$, is the length of the shortest path between them in $G$. We omit the subscript $G$ if it is evident. For a point $x \in A(G)$ and a vertex $v \in V(G)$, the \emph{weighted distance} between the point $x$ and the vertex $v$ is equal to $w(v) \cdot d(x, v)$. Given a set of points $X$, the \emph{cost} of covering the vertices of $G$ by $X$ is given by the expression

\begin{equation*}
    \max_{v \in V(G)} \left \{\min_{x \in X} \left \{w(v) \cdot d(x, v) \right \} \right \}
\end{equation*}

The objective of the \emph{weighted $k$-center} problem in $G$ is to find a set of $k$ points $X$ such that the cost of covering $V(G)$ is minimized. This cost is also called the \emph{optimal cost} of the weighted $k$-center problem in $G$.

We say a point (or a center) $x \in A(G)$ \emph{covers} a vertex $v \in V(G)$ with cost $\lm,$ if $w(v) \cdot d(x, v) \leq \lm$. Extending this definition we say a set of centers $X$ covers $V(G)$ with cost $\lm$ if for all vertices $v \in V(G)$ its nearest center in $X$ covers $v$ with cost $\lm$. We sometimes do not explicitly mention $\lm$ if it is well understood.

\subsection{Tree representation of a cactus}

Given a cactus $G$, a vertex of $G$ is a called a \emph{graft} if it does not belong to any cycle of $G$. It is called a \emph{hinge} if it belongs to a cycle of $G$ and has degree at least 3. Each of the remaining vertices, if any, belongs to exactly one cycle of $G$.

We transform $G$ into a new graph $T$, where the nodes of $T$, $V(T)$, represent either grafts, hinges or cycles of $G$. If a node $t_i \in V(T)$ represents a cycle of $G$ then we add an edge to every node $t_j \in V(T)$ which represents a hinge of that cycle. If $t_i, t_j \in V(T)$ both do not represent a cycle then they are adjacent if and only if the vertices of $G$ that they represent are adjacent.

Observe that two nodes of $V(T)$ representing a graft and a cycle of $G$, respectively, are never adjacent. Since $G$ is a cactus, the graph $T$ thus formed is a tree. We call $T$ the \emph{tree representation} of $G$ \cite{ben}. See Figure \ref{fig:cactus-tree} for an illustration.

\begin{figure}[ht]
    \centering
    \subfloat[A cactus $G$. \label{fig:cactus}]{
        \includegraphics[width=0.5\textwidth]{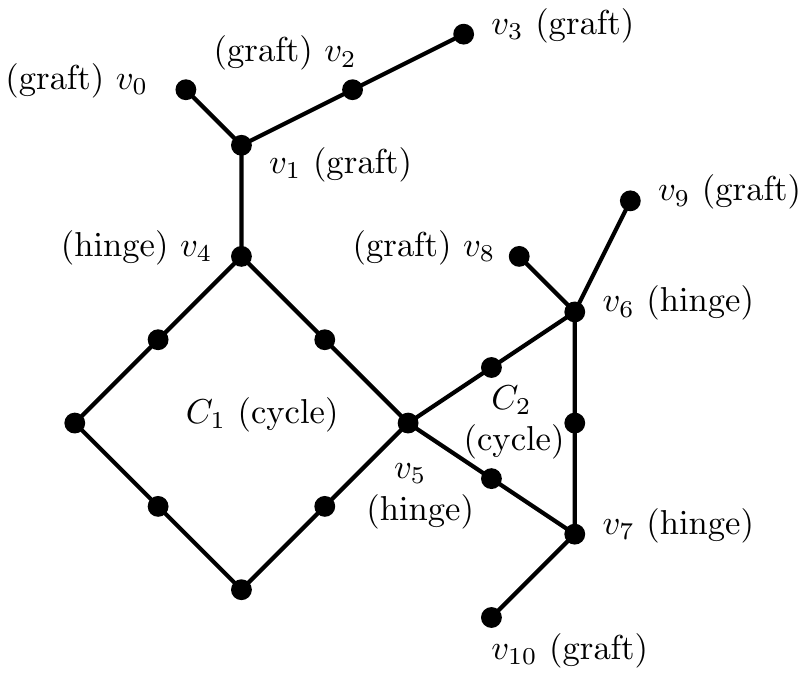}
    }
    \hfill
    \subfloat[Tree representation $T$ of $G$ rooted $t_0$. For all $t \in V(T)$, the vertex $t$ represents is mentioned inside brackets. \label{fig:tree}]{
        \includegraphics[width=0.45\textwidth]{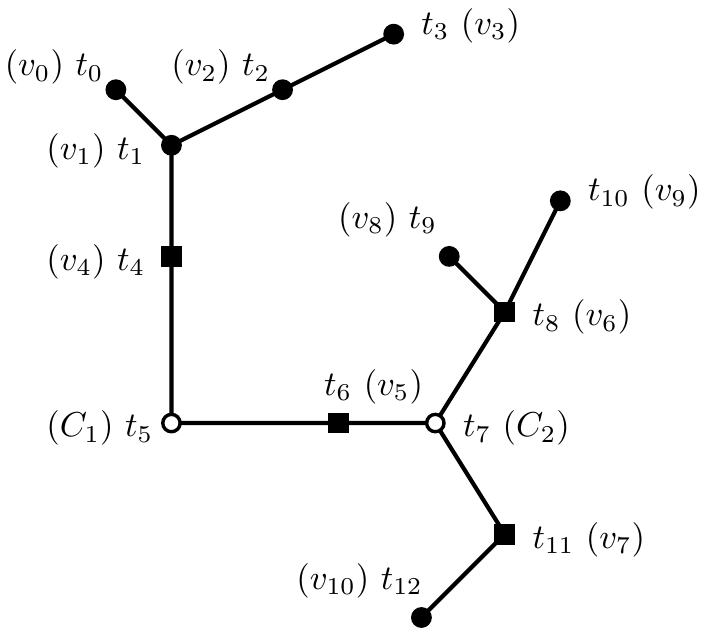}
    }

    \caption{The associate vertex of $t_0$ is $v_0$ and that of $t_4$ and $t_5$ are both $v_4$. \label{fig:cactus-tree}}
\end{figure}

Let $t_r \in V(T)$ be an arbitrary leaf node that we designate as the root of $T$. This establishes a parent child relation on the nodes in $V(T)$. We now map each node $t \in V(T)$ to a vertex $v \in V(G)$. We call $v$ the \emph{associate vertex} of $t$.

For each $t_i \in V(T)$ which represents either a graft or a hinge $v_i \in V(G)$, we assign $v_i$ as the associate vertex of $t_i$. For each $t_j \in V(T)$ which represents a cycle of $G$, let $t_j'$ be its parent. Since $t_j$ represents a cycle, $t_j'$ will always represent a hinge of $G$. The associate vertex of $t_j$ is set as the associate vertex of $t_j'$. If $t_j$ represent a cycle of $G$ and is also the root ($t_j = t_r$), we arbitrarily select a vertex of that cycle and set it as the associate vertex of $t_j$. We call the the associate vertex of $t_r$, say $v_r$, the root of $G$.

\begin{lemma}[Ben-Moshe et al. \cite{ben}]
    Given a cactus $G$ with $|V(G)| = n$, its tree representation $T$ can be found out in $O(n)$ time.
\end{lemma}

\subsection{Circular arcs} \label{sec:cacrs}

A \emph{circular arc} is the arc of a circle between a pair of distinct points. In this section, we describe some operations on a set of circular arcs. Let $C$ be a set of $m$ circular arcs. We assume the midpoints of all arcs of $C$ to be distinct and sorted in the clockwise direction. Let $c_1, c_2, \ldots, c_m$ be the given sorted order of $C$. For $c_i \in C, i < m$, $c_{i + 1}$ is called the \emph{clockwise neighbour} of $C$. For $c_m$, $c_0$ is its clockwise neighbour. The \emph{counterclockwise neighbour} is defined analogously.

We say an arc $c \in C$ is a \emph{super arc} if there exists another arc $c' \in C$ which it completely contains, i.e. $c \supset c'$. We now prove the following.

\begin{lemma} \label{prop:superarcs}
    All super arcs in $C$ can be removed in $O(m)$ time.
\end{lemma}

\begin{proof}
    Mark all arcs in $C$ as not visited. Arbitrarily select a pair of consecutive arcs $a, b \in C$. Note that $b$ is the clockwise neighbor of $a$. While $a$ is not marked visited we repeat the following steps.
    
    If $a$ is contained in $b$, replace $b$ by its clockwise neighbor and remove old $b$ from $C$. Else, mark $a$ as visited and update both $a$ and $b$ to their respective clockwise neighbors.

    Repeat the entire process once more in the counterclockwise direction.

    The above steps remove all super arcs in $C$. Since, all arcs in $C$ are traversed once in the clockwise direction and once in the counterclockwise direction, the algorithm takes $O(m)$ time. \qed
\end{proof}

The pseudocode for the algorithm of Lemma \ref{prop:superarcs} is presented as Algorithm \ref{algo:superarcs}.

\begin{algorithm}[ht]
    \caption{Remove Super Arcs} \label{algo:superarcs}
    \begin{algorithmic}[1]
        \State Mark all arcs in $C$ as unvisited

        \State Let $a \in C$
        \State Set $b \leftarrow$ the clockwise neighbor of $a$

        \While{$a$ has not be marked visited}
            \If{$a \subset b$}
                \State Set $b' \leftarrow$ the clockwise neighbor of $b$
                \State Remove $b$ from $C$
                \State Set $b \leftarrow b'$
            \Else
                \State Mark $a$ as visited
                \State Set $a \leftarrow b$
                \State Set $b \leftarrow$ the clockwise neighbor of $b$
            \EndIf
        \EndWhile

        \State Repeat the entire process once more in the counterclockwise direction



    \end{algorithmic}
\end{algorithm}

A set of circular arcs is called \emph{proper} if it does not contain any super arcs. For the remaining part of this section, we consider $C$ to be a set of proper circular arcs. The \emph{minimum piercing set} $\chi$ of $C$ is a set of points of minimum cardinality such that each arc of $C$ contains some point of $\chi$. The cardinality of $\chi$ is called the \emph{piercing number} of $C$. Hsu et al. \cite{hsu} showed that the maximum independent set of a set of $m$ circular arcs can be found in $O(m)$ time. The same technique can be used to show the following result.

\begin{lemma}[Hsu et al. \cite{hsu}]
    The piercing set of $C$ can be found in linear time when the midpoint of the arcs in $C$ are given in sorted order.
\end{lemma}

Each arc $c \in C$ has two endpoints. Starting from the interior of $c$, the endpoint encountered first in a clockwise traversal is the \emph{clockwise endpoint} of $c$. The other endpoint is the \emph{counterclockwise endpoint} of $c$.
The \emph{clockwise next} of $c$ is the arc $c'$ which is encountered first in a clockwise traversal starting at the clockwise endpoint of $c$ and which does not intersect $c$. We similarly define the \emph{counterclockwise next} of $c$. See Figure \ref{fig:clockwisenext} for reference.

\begin{figure}[ht]
    \centering
    \includegraphics[width = .45 \textwidth]{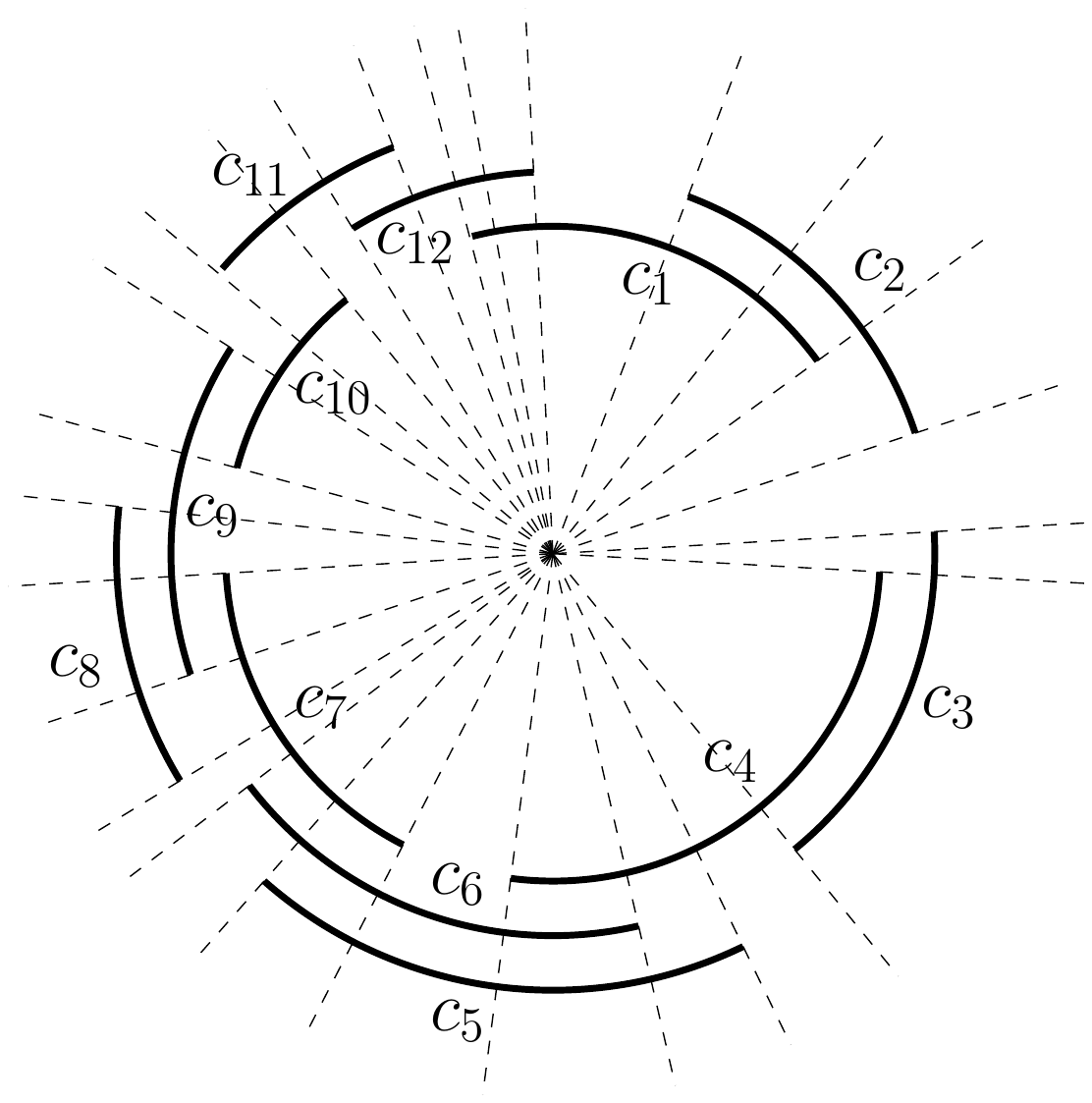}
    \caption{For the circular arc $c_1$, its clockwise neighbour is $c_2$, its clockwise next is $c_3$ and for say $q = 3$ its clockwise $q$-next is $c_8$.} \label{fig:arcs}
    \label{fig:clockwisenext}
\end{figure}

Since the endpoints of $C$ are given to us in sorted order, we make the following claim.

\begin{lemma} \label{prop:next}
    The clockwise (counterclockwise) next of all arcs in $C$ can be found in $O(m)$ time.
\end{lemma}

\begin{proof}
    Mark all arcs in $C$ as visited. Arbitrarily select an arc $a \in C$ and set $b \leftarrow a$. While $a$ is not marked visited repeat the following.
    
    If $a$ intersect $b$, update $b$ to its clockwise neighbour. Else, mark $a$ as visited and set $b$ as the clockwise neighbour of $a$. Update $a$ to its clockwise neighbour.

    The above steps set the clockwise next for all arcs in $C$ in $O(m)$ time. The counterclockwise next of all arcs can be found similarly. \qed
\end{proof}

The pseudocode for the algorithm of Lemma \ref{prop:next} is presented as Algorithm \ref{algo:next}. The algorithm to find all counterclockwise nexts of the arcs of $C$ is similar.

\begin{algorithm}[ht]
    \caption{Clockwise Next} \label{algo:next}
    \begin{algorithmic}[1]
        \State Mark all arc in $C$ as not visited
        
        \State Let $a \in C$
        \State Set $b \leftarrow a$
        
        \While{$a$ is not marked visited}
            \If{$a \cap b \neq \varnothing$}
                \State Set $b \leftarrow$ the clockwise neighbour of $b$
            \Else
                \State Set $b$ as the clockwise next of $a$
                \State Mark $a$ as visited
                \State Set $a \leftarrow$ the clockwise neighbour of $a$
            \EndIf
        \EndWhile
    \end{algorithmic}
\end{algorithm}

For any positive integer $q$, let $[q]$ denote the set $\{ 1, 2, \ldots, q \}$. We now define the \emph{clockwise $q$-next} of an arc $c_0 \in C$. Let $c_0, c_1, \ldots, c_q$ be a sequence of circular arcs of $C$ such that for any $i \in [q]$, $c_{i}$ is the clockwise next of $c_{i-1}$. The arc $c_q$ then is called the clockwise $q$-next of $c_0$. Note that, for any, $i \in [q]$, the arc $c_{i}$ is the clockwise $i$-next of $c_0$. The counterclockwise $q$-next of $c$ is defined similarly. See Figure \ref{fig:arcs} for reference.

\begin{lemma} \label{prop:qnext}
    The clockwise (counterclockwise) $q$-next of all arcs of $C$ can be found in $O(m)$ time.
\end{lemma}

\begin{proof}
    The pseudocode for Lemma \ref{prop:qnext} is presented as Algorithm \ref{algo:qnext}. The algorithm uses the compatibility forest data structure for intervals in Gavruskin et al. \cite{gavruskin} by extending it to circular arcs and appropriately naming it the compatibility graph data structure.
    
    Let $G_C$ represent the clockwise compatibility graph of $C$ which is a graph with vertices for each arc of $C$ and an edge from $a$ to $b$, $a, b \in C$, if $b$ is the clockwise next of $a$. The graph $G_C$ can be disconnected with each connected component having exactly one directed cycle. Note that there are also no other undirected cycle in that component.

    A careful implementation of Algorithm \ref{algo:qnext} will result in an $O(m)$ time algorithm. The algorithm for the counterclockwise case is similar. \qed
\end{proof}

\begin{algorithm}[ht]
    \caption{Clockwise $q$-Next} \label{algo:qnext}
    \begin{algorithmic}[1]
        \State Initialize all arcs of $C$ as unvisited
        \For{each connected component $G'$ of $G$}
            \State Let $r$ be an arc on the directed cycle of $G'$
            \State Remove the edge connecting $r$ to its clockwise next to convert $G'$ into a tree $T'$ rooted at $r$
            \State Set $a \leftarrow r$
            \State Initialize a dequeue $Q$ of size $q$ such that the $i^{\text{th}}$ element in $Q$ is the clockwise $i$-next of $a$
            \While{$r$ is not marked visited}
                \If{$a$ has an unvisited child $b$}
                    \State Update $Q$ by inserting $a$ as its first element and deleting its last element
                    \State Set $a \leftarrow b$
                \Else
                    \State Let $b$ be the last element in $Q$
                    \State Set $b$ as the clockwise $q$-next of $a$ and marking $a$ visited
                    \State Update $Q$ by removing its first element and inserting the clockwise next of $b$ as its last element
                    \State Set $a \leftarrow $ the parent of $a$
                \EndIf
            \EndWhile
        \EndFor
    \end{algorithmic}
\end{algorithm}

\section{Feasibility Test} \label{sec:feasibility}

In this section, we present an algorithm which solves the decision version of the weighted $k$-center problem in time linear in the size of the cactus. The problem is commonly referred to as the \emph{feasibility test} in literature. Given a cactus network $G$, let $\lm^* \in \R_{\geq 0}$ be the optimal cost of the weighted $k$-center problem in $G$. For an input $\lm \in \R_{\geq 0}$, the feasibility test asks whether $\lm \geq \lm^*$, or not. Note that the optimal cost $\lm^*$ is an unknown value. If $\lm \geq \lm^*$, then we say $\lm$ is feasible, otherwise, $\lm < \lm^*$ and we say $\lm$ is infeasible.

For a given $\lm$, our algorithm works by placing a minimum number of centers $X$, such that every vertex in $G$ is covered by some center in $X$ with cost $\lm$. If $|X| \leq k$, then $\lm$ is feasible; otherwise $|X| > k$, and $\lm$ is infeasible. The algorithm uses the tree representation $T$ of $G$ in order to fix a processing order on the vertices of $G$. Recall that $T$ is rooted at the node $t_r$. We process the nodes of $T$ in a bottom-up order, visiting the leaves first, and then visiting the nodes whose children have been processed.

\begin{figure}[ht]
    \centering
    \includegraphics[width = \textwidth]{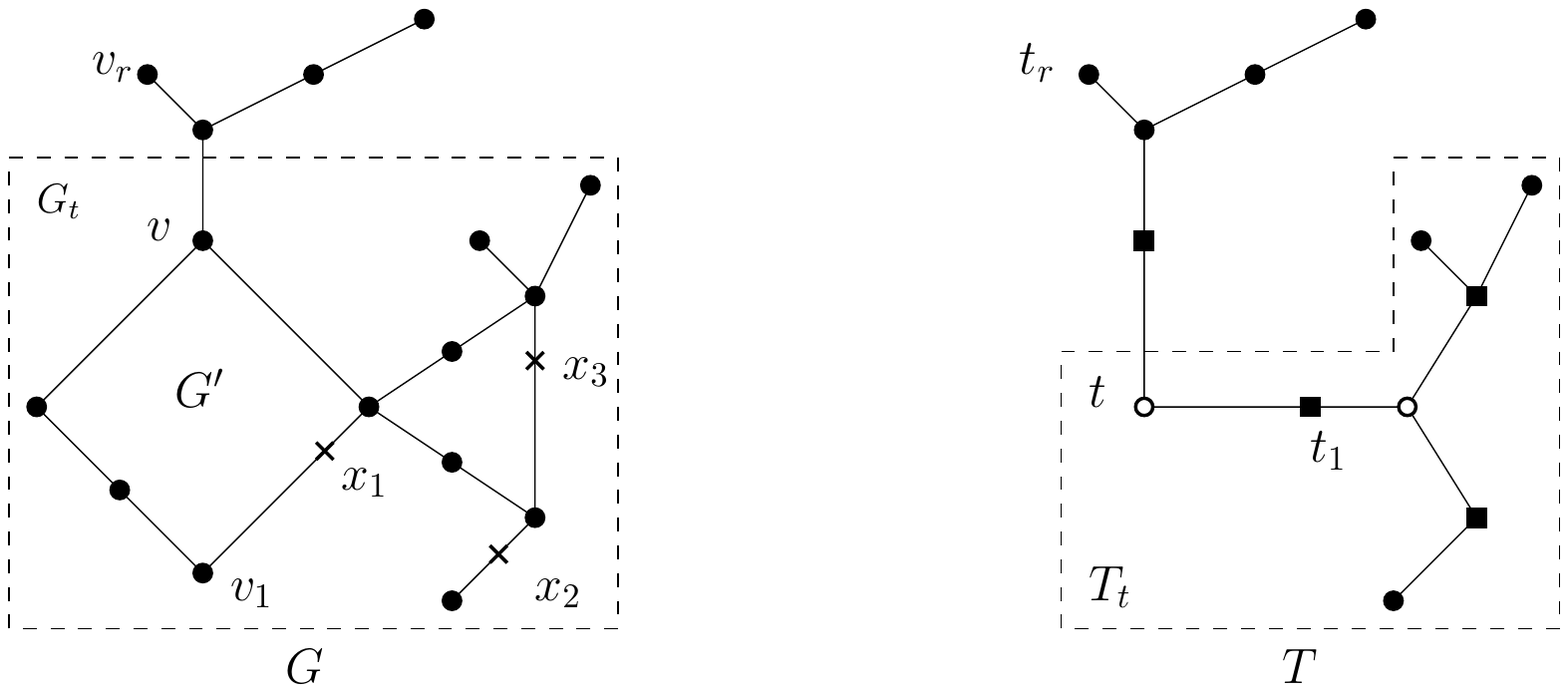}
    \caption{The subgraphs $T_t$ and $G_t$. Vertex $v$ and $v_1$ in $V(G)$ are the associate of $t$ and $t_1$ in $V(T)$, respectively. Node $t_1$ is a child of $t$, $t \in V(T)$ represents the cycle $H$ of $G$, and $X_t = \{x_1, x_2, x_3\}$.}
    \label{fig:ftest}
\end{figure}

For $t \in V (T)$, let $T_t$ denote the subtree of $T$ rooted at $t$ and let $G_t$ denote the subgraph of $G$ that $T_t$ represents. For each node $t \in V (T)$, we store two values $\cl(t)$ and $\cv(t)$, defined as follows.

Consider an intermediate step in our feasibility test algorithm where all vertices in $G_t$ have already been processed. The algorithm may have placed some centers inside $G_t$. Let these centers be represented by the set $X_t$. The \emph{uncovered} vertices of $G_t$ are those vertices which cannot be covered by $X_t$ with cost $\lm$.

Let $v_t \in V(G)$ be the associate vertex of $t \in V(T)$. The function $\cv(t)$ represents the maximum distance that a center can be placed outside $G_t$ from $t$, that covers all uncovered vertices of $G_t$ with cost $\lm$. Let $V' \subseteq V(G_t)$ be the uncovered vertices of $G_t$. We define $\cv(t)$ as follows.

\begin{align*}
    \cv(t) = \begin{cases}
        \min \left \{ \lm/w(v) - d(v, v_t) \mid \forall v \in V' \right \} & \text{ if } V' \neq \varnothing \\
        \infty & \text{ otherwise }
    \end{cases}
\end{align*}

The function $\cl(t)$ stores the distance of the closest center inside $G_t$ from $v_t$. We define $\cl(v)$ as follows.

\begin{align*}
    \cl(t) = \begin{cases}
        \min \left \{ d(x, v_t) \mid \forall x \in X_t \right \} & \text{ if } X_t \neq \varnothing \\
        \infty & \text{ otherwise }
    \end{cases}
\end{align*}

Recall that we traverse the nodes of $T$ bottom-up. We also maintain a variable $count$, which counts the number of centers placed till now. Let $t$ be the node of $T$ that we are currently at and let $v$ be its associate vertex. In Section \ref{sec:hinge_graft} and \ref{sec:cycle} we show how to place centers inside $G_t$ depending on whether $t$ represents a hinge/graft or a cycle of $G$.

\subsection{$t$ represents a hinge or a graft of $G$} \label{sec:hinge_graft}

Let $V_t \subseteq V(T)$ be the set of children of $t$. For each $t_i \in V_t$, let $v_i$ be its corresponding vertex.

Let $V_1 \subseteq V_t$, be such that $\forall t_i \in V_1$, $\cv(t_i) < d(v_i, v_t)$. Observe that for all such $t_i \in V_1$, its is necessary to place a center on the edge $(v_i, v_t)$ in order to cover the remaining uncovered vertices in $G_{t_i}$. We do so by placing a center on the edge $(v_t, v_i)$ at a distance $\cv(t_i)$ from $v_i$ and incrementing $count$ by one. Since all vertices in $T_{t_i}$ are now covered, we set $\cv(t_i)$ to $\infty$.

As per definition, we now set $\cl(t)$ as the minimum of all distances of all centers placed in $G_t$, from $v_t$. Let $x$ be such a center i.e. $x$ is closest to $v_t$. If $G_t$ has no centers placed, we set $\cl(t)$ to $\infty$.

Now, let $V_2 \subseteq V_t \setminus V_1$ be such that $\forall t_j \in V_2$, $\cv(t_j) \geq d(v_t, v_j) + \cl(t)$. Observe that for all such $t_j \in V_2$, all vertices in $G_{t_j}$ are either covered by the centers placed inside $G_{t_j}$ or by the center $x$, if it exists. Therefore, we set $\cv(t_j)$ as $\infty$.

Finally, we update $\cv(t)$. We define the parameter $\delta$ as
\begin{equation*}
    \delta = \min \{\cv(t_i) - d(v_t, v_i) \mid \text{for all child } t_i \text{ of } t\}
\end{equation*}

If $\cl(t) \leq \frac{\lm}{w(v_t)}$ then we set $\cv(t) = \delta$, otherwise $\cv(t) = \min \{ \lm / $ $w(v_t), \delta \}$.

\subsection{$t$ represents a cycle of $G$} \label{sec:cycle}

Let $H$ be the cycle of $G$ that $t$ represents and let $m = |V(H)|$. For a vertex $v_j \in V(H)$, let $\nr(v_j)$ be the distance of the nearest center in $G_t$ to $v_j$. The definition of $\nr(\cdot)$ is similar to that of $\cl(\cdot)$. Where $\cl(\cdot)$ is defined on the nodes in $V(T)$, $\nr(\cdot)$ instead is define on the vertices of the cycles in $G$. Observe that the only information required to compute $\nr(v_j), \forall v_j \in V(H)$ is $\cl(t_i), \forall t_i \in V_t$. Therefore, it is not difficult to show that $\forall v_j \in V(H), \nr(v_j)$ can be computed in $O(m)$ time.

If $v_j \neq v_t$ then $v_j$ is uncovered if $\nr(v_j) > \cv(t_j)$, where $t_j \in V(T)$ represents $v_j$. If $v_j = v_t$ or is not a hinge then $v_j$ is uncovered if $\nr(v_j) > \lm / w(v_j)$. Our objective is to cover all the uncovered vertices in $H$ by placing zero or more centers in $H$ and at most one outside $G_t$.

We convert this problem of placing centers in $H$ to an equivalent piercing set problem on a set of circular arcs $C$ as follows. For $i \in [m]$, if $v_i$ is a hinge and if $v_i \neq v_t$, let $t_i \in V(T)$ be the node that represents $v_i$. We set $l_i = 2 \times \cv(t_i)$. Otherwise, we set $l_i = 2 \lm / w(v_i)$. We insert into $C$ a circular arc of length $l_i$ centered at $v_i$.


Our objective is to place a set of piercing points so that the overall number of centers placed is minimum. Let $C' \subseteq C$ be the subset of circular arcs which do not intersect $v_t$. Let $q$ and $q'$ be the size of the minimum piercing set of $C$ and $C'$, respective. Observe that $q - 1 \leq q' \leq q$. If $q' = q - 1$, then instead of placing a center $x$ at $v_t$, it can be shifted outside $G_t$, within a permissible distance.

We therefore find a set of piercing points $\chi$ such that either (a) all uncovered vertices of $H$ are covered by $\chi$, and $\exists x \in \chi$ which is placed as close as possible to $v_t$; (b) not all uncovered vertices of $H$ are covered by $\chi$ and one more center $x$ is needed to be placed outside $G_t$ to cover them. In the later case, the $\chi$ is selected such that $x$ can be placed as far away as possible from $v_t$. If $q' = q$, then Case (a) follows, otherwise $q' = q - 1$ and Case (b) follows. Given $H$ in cyclic order, we show how to compute the set $\chi$ in $O(m)$ time.


\begin{figure}[ht]
    \centering
    \includegraphics[width = .5 \textwidth]{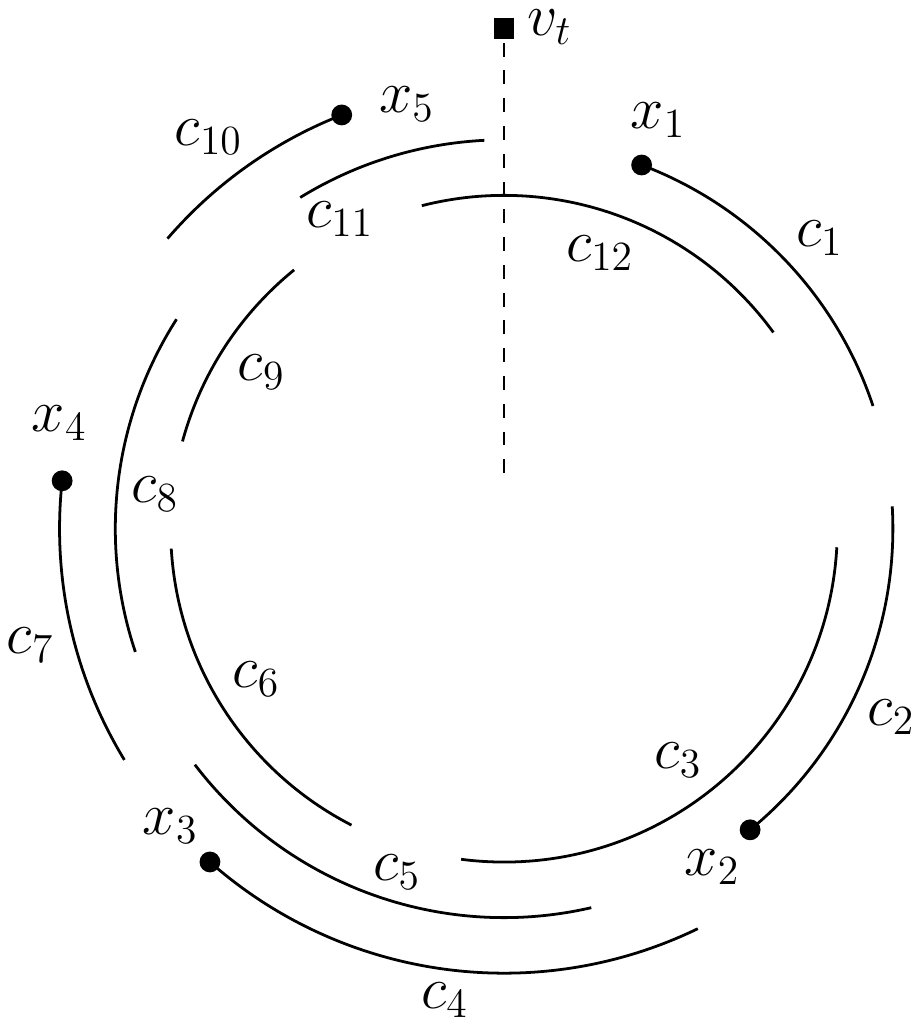}
    \caption{The set $\chi(c_1)$ for $q$/$q'$ = 5.}
\end{figure}

\subsubsection{Case (a)}

In this case, we need to find a minimum piercing set $\chi$ of $C$ such that one of the piercing points $x \in \chi$ is closest to $v_t$. We check for all possible configurations of $\chi$ and chose the one with the closest point to $v_t$. Let $c_1$ be the first arc in the clockwise direction from $v_t$ which does not overlap with it. Let $c_j$ be the clockwise next of $c_1$. For each $c_i \in \{c_1, c_2, \ldots, c_{j-1}\}$, we construct a set of piercing points $\chi(c_j)$ as follows.

We place the first piercing point, $x_1$, at the counterclockwise endpoint of $c_i$. We place the next piercing point, $x_2$, at the clockwise endpoint of the clockwise neighbour of $c_i$, say $c_{i + i_1'}$. We place the next piercing point, $x_3$, at the clockwise endpoint of the clockwise next of $c_{i + i_1'}$, say $c_{i + i_2'}$. We continue placing points in this way and place the $q$-th piercing point $x_{q}$ on the clockwise endpoint of the clockwise $(q-1)$-next of $c_i$, say $c_{i + i_{q - 1}'}$. The set $\chi(c_i)$ is a valid minimum piercing set of $C$ if the clockwise next of $c_{i + i'_{q-1}}$ overlaps with $x_1$.

Note that for any $c_i \in \{c_1, c_2, \ldots, c_{j-1}\}$ the closest point of $\chi(c_i)$ to $v_t$ is either $x_1$ or $y_{q}$. Therefore, we are happy with only computing the first and the $q$-{th} piercing point for each $\chi(c_i)$, which can be done in $O(1)$ time. We select $\chi$ as one of the $\chi(c_i)$ which is a valid  piercing set and which has a point closest to $v_t$. We compute each point in $\chi$ and place a center in $H$ at its corresponding position. We increment $count$ by $q$ and update $\cl(t)$ and $\cv(t)$ accordingly.

\subsubsection{Case (b)}

In this case, we find a set of $q'$ centers on $H$, such that the center outside $G_t$, covering the remaining uncovered vertices in $G_t$ through $v_t$, can be as far away as possible from $v_t$. For each $c_i \in \{c_1, c_2, \ldots, c_{j-1}\}$, we construct a set of piercing points $\chi(c_j)$ as shown earlier where $q$ is replaced by $q'$.

As noted earlier, for any $c_i \in \{c_1, c_2, \ldots, c_{j-1}\}$ the closest point of $\chi(c_i)$ to $v_t$ is either $x_1$ or $x_{q'}$. Therefore, we are happy with only computing the first and the $q'$-{th} piercing point for each $\chi(c_i)$, which can be done in $O(1)$ time.

Let $c_{i - i'_0}$ be the counterclockwise next of $c_i$ and $c_{i + i'_{q'}}$ the clockwise next of $c_{i + i'_{q'-1}}$ and let $v_1$ and $v_2$ be the vertices of $H$ that they respectively represent. Let $V'$ be the set of centers in $G_t$ that are uncovered after placing the centers corresponding to $\chi(c_i)$. Then the distance of the farthest center from $v_t$ which covers the remaining uncovered vertices of $G_t$ from outside is given by $\min \{ \alpha(v_1) - d(v_1, v_t), \alpha(v_2) - d(v_2, v_t) \}$, where $\alpha(v_l)$ ($l = \{1, 2\}$) is $\cv(v_l)$, if $v_l$ is a hinge, and $\lm/v_l$, otherwise. Let us denote this value as the \emph{cost} of $\chi(c_i)$.

We select $\chi$ as one of the $\chi(c_i)$ which has the largest cost. We compute each point in $\chi$ and place a center in $H$ at its corresponding position. We increment $count$ by $q'$ and update $\cl(t)$ and $\cv(t)$ accordingly. If $t = t_r$, we place one more center at $v_r$ and increment $count$ by 1.

\subsection{Time complexity}

\begin{theorem}
    The feasibility test for the weighted $k$-center problem on cactus takes $O(n)$ time, where $n$ is the number of vertices in the cactus.
\end{theorem}

\begin{proof}
    We traverse each node of $T$ only once. If the node $t$ represents either a hinge or a graft, then the time required to process it is linear in the number of children of $t$. Otherwise, $t$ represents a cycle, and we convert our problem to a problem on circular arcs. This step is again linear in the number of vertices in the cycle. Finally, finding $\chi, \cl(t)$ and $\cv(t)$ also take time linear in the number of vertices in the cycle. Therefore, the running time of the entire algorithm is $O(n)$.  \qed
\end{proof}

\section{Parametric Pruning on Paths} \label{sec:arrangements}

Let $L$ be a set of $m$ lines in $\R^2$ and let $I$ represent the set of $n \choose 2$ pairwise intersection points of the lines in $L$. Consider a feasibility test, similar to the one define in Section \ref{sec:feasibility}, which takes as input a value $\gamma$ and returns \emph{feasible} if $\gamma \geq \gamma^*$, where $\gamma^*$ is a predefined threshold. It returns \emph{infeasible} otherwise.

Let $p_1$ is the highest point in $I$ whose $y$-coordinate value is infeasible and let $p_2$ be the lowest point in $I$ whose $y$-coordinate value is feasible. We have the following result by Chen and Wang \cite{chen}.

\begin{lemma}[Chen and Wang \cite{chen}] \label{lem:lines}
    Given $L$, both $p_1$ and $p_2$ can be found in $O \left((m + \tau) \log m \right)$, where $\tau$ is the running time of the feasibility test.
\end{lemma}

Let $P$ be a path of size $m$ with weights on its vertices and lengths on its edges as shown in Section \ref{sec:intro}. For $u, v \in V(P)$, let $\ct(u, v)$ represent the cost of covering $u$ and $v$ by a single center in $A(P)$ and let $\lct(u, v)$ represent the cost of covering $u$ by a center placed at $v$. Let $R_0(P) = \{\lct(u, v) \mid \forall (u, v) \in V(P) \times V(P)\}$, $R_1(P) = \{\ct(u, v) \mid \forall (u, v) \in V(P) \times V(P)\}$ and $R(P) = R_0(P) \cup R_1(P)$.

Let $v_0 \in V(P)$ be the leftmost vertex of $P$. For each vertex $v_i \in V(P)$, we place a point $q_i$ on the $x$-axis of an $xy$-coordinate system at $x$-coordinate $d(v_0, v_i)$. We place for every vertex $v_i \in V(P)$, three lines $l^+_i, l^-_i$ and $l^{\infty}_i$, with slopes $+w(v_i), -w(v_i)$ and $\infty$ (vertical line), respectively, at point $q_i$. Let $L_P = \left\{l^+_i, l^-_i, l^{\infty}_i \mid \forall v_i \in V(P) \right\}$ and let $I_P$ represent the set of pairwise intersection points of the lines in $L_P$.

\begin{figure}[ht]
    \centering
    \includegraphics[width = .8 \textwidth]{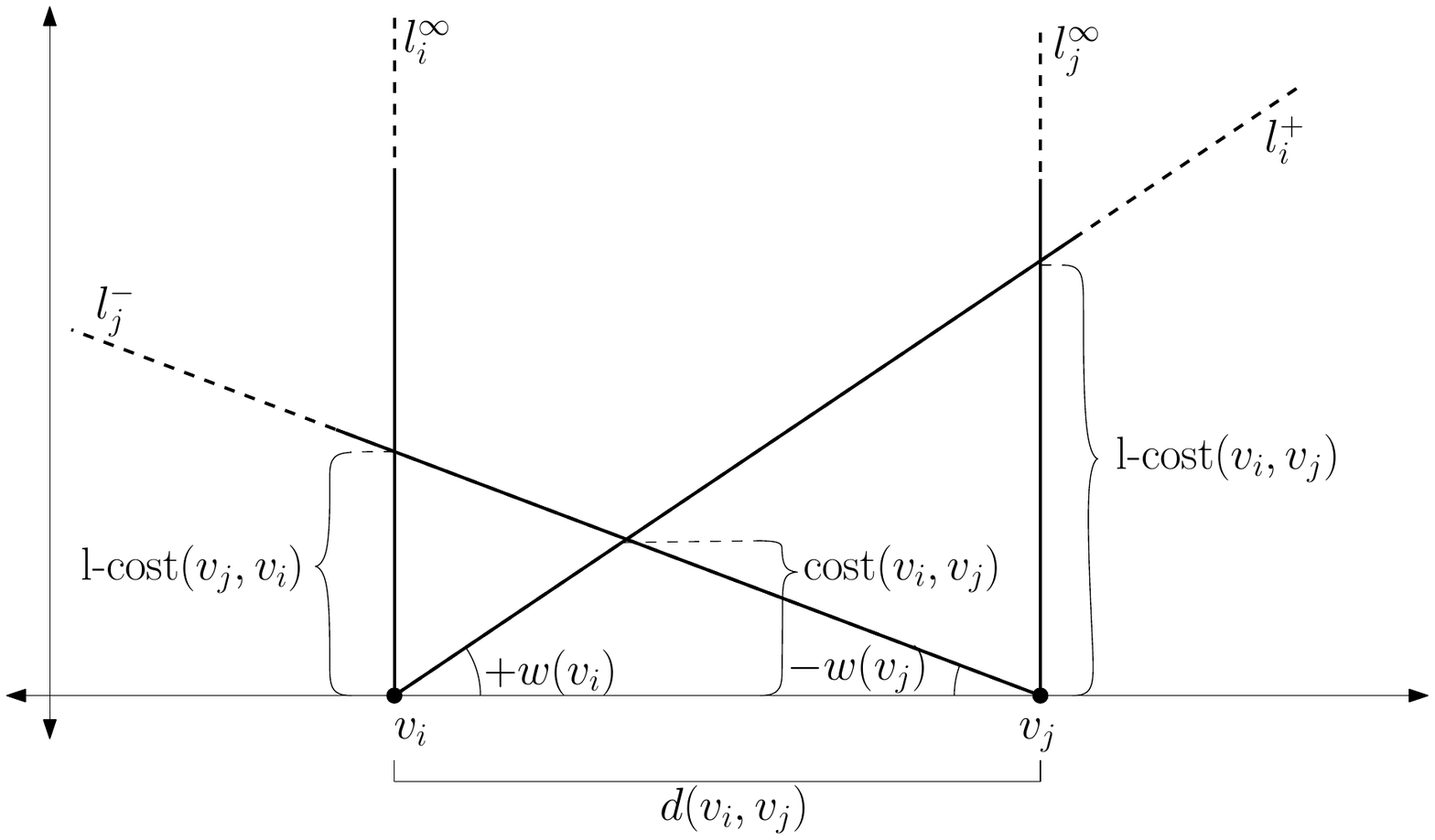}
    \caption{Representing values of $R_0(P) \cup R_2(P)$ in $\R^2$.}
    \label{fig:costs}
\end{figure}

Observe that all values in $R(P)$ are represented as the $y$-value of the points in $I_P$. See Figure \ref{fig:costs} for reference. Therefore, with respect to a feasibility test with predefined threshold $\gamma^*$, we can use Lemma \ref{lem:lines} to find $\lm_1, \lm_2 \in \R_{\geq 0}$ such that $\lm_1 < \gamma^* \leq \lm_2$ and no value in $R(P)$ lies inside the range $(\lm_1, \lm_2)$. The same result can be achieved if instead of a single path we have multiple paths. Given multiple paths $P_1, P_2, \ldots, P_l$ we concatenate the paths by joining their end points sequentially to get a path $\P$. We can then perform the same operations on this resultant path to get $\lm_1, \lm_2 \in \R_{\geq 0}$ such that $\lm_1 < \gamma^* \leq \lm_2$ and $\forall \lm \in \bigcup_{i = 1}^l R(P_i), \lm \notin (\lm_1, \lm_2)$. We therefore make the following claim.

\begin{lemma} \label{lem:paths}
    Given a collection of path networks $P_1, P_2, \ldots, P_l$, and a feasibility test with threshold $\gamma^*$, we can find $\lm_1$ and $\lm_2$ in $O((m + \tau) \log m)$ where $m$ is the size of $\P$ and $\tau$ is the running time of the feasibility test.
\end{lemma}

\section{$k$-Center on Cactus Stems} \label{sec:stem}

In this section, we define a \emph{cactus stem}, which is a generalization of the (path) stem defined by Frederickson \cite{frederickson1991optimal} for trees and present two algorithm to compute the weighted $k$-center on it. Let $S'$ be a cactus whose tree representation is a path $P'$. Assume the nodes of $P'$ to be ordered from left to right.

The structure of $S'$ is an alternating sequence of paths and cycles. Let $P_1, C_1, P_2, C_2, \ldots, P_m, C_m$ be the left to right ordering of this sequence with respect to the left to right ordering of $P'$, where $P_i$ and $C_j$ represent a path and a cycle of $S'$, respectively. Note that the path $P_1$ and cycle $C_m$ may not exist.

We call a vertex a \emph{pendant vertex} if its degree is 1. Let $S$ be the graph generated from $S'$ by attaching any number of pendant vertices to the vertices in $V(S')$. We call $S$ a \emph{cactus stem}. For this section we assume $|V(S)| = n$.

Let $s$ and $t$ be the leftmost and rightmost nodes of $P'$ and let $g_0$ and $h_0$ be their respective associate vertices in $S'$. We call $g_0$ and $h_0$ the \emph{leftmost} and \emph{rightmost} vertices of $S$, respectively. Let $g_i$ and $h_i$ represent the left and right hinges of a cycle $C_i$ in $S$. If the left hinge of $C_1$ does not exist then we set $g_1$ as $g_0$. Similarly, if the right hinge of cycle $C_m$ does not exist we set $h_m$ as $h_0$.

The \emph{minor hemisphere} of $C_i$ is the shortest path between $g_i$ and $h_i$ in $C_i$. The \emph{major hemisphere} of $C_i$ is the other path between $g_i$ and $h_i$ in $C_i$. If both paths are of the same length, then we arbitrarily select one of them as the major hemisphere and the other as the minor hemisphere. See Figure \ref{fig:cactusstem}.

\begin{figure}[ht]
    \centering
    \includegraphics[width = \textwidth]{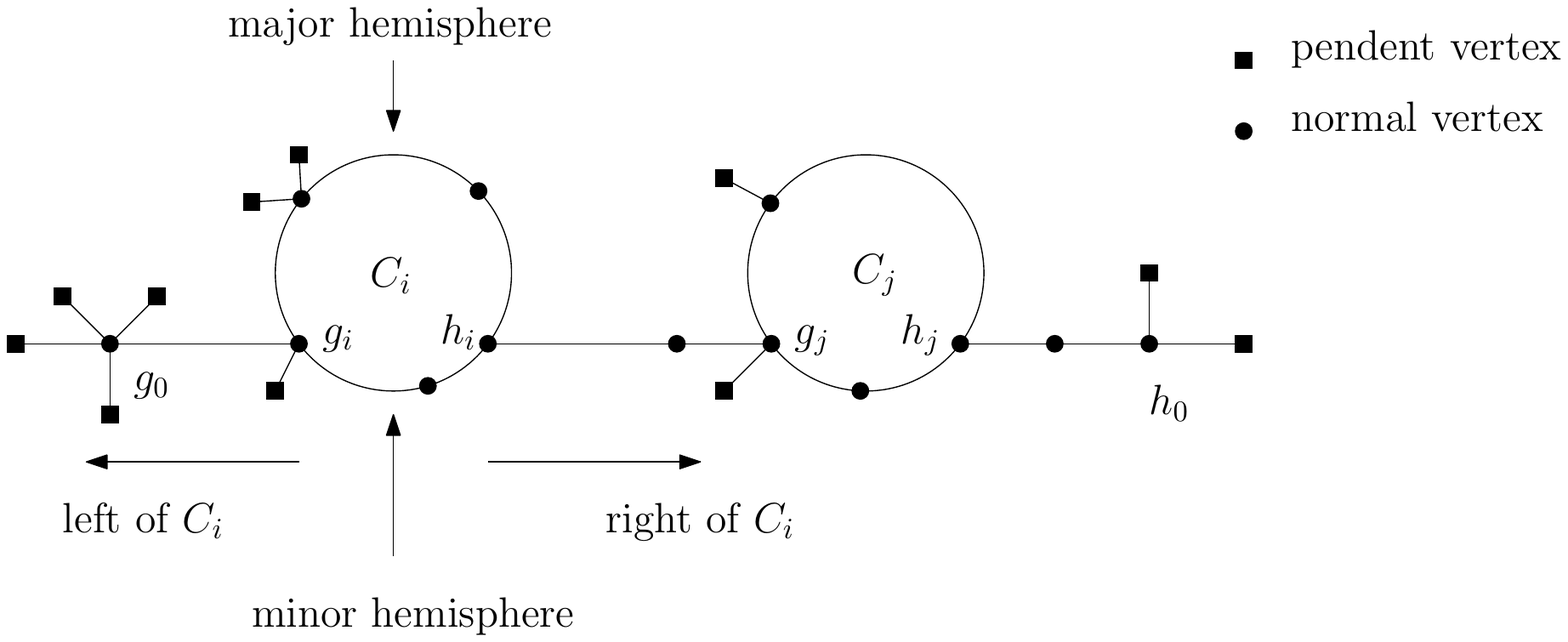}
    \caption{A cactus stem.} \label{fig:cactusstem}
\end{figure}

A vertex $u \in V(S)$ is said to be to the \emph{left} of $C_i$ if no path between $u$ and $g_0$ passes through $C_i$. Similarly, a vertex $v \in V(S)$ is said to be to the \emph{right} of $C_i$ if no path between $v$ and $h_0$ passes through $C_i$. Let $\C_i$ be the graph formed by attaching all pendant vertices adjacent to $V(C_i)$ to $C_i$. Observe that $u, v \notin V(\C_i)$.

Recall that $A(S)$ represents the set of all points on all edges of $S$. For $x, y \in A(S)$, let $\pi_{xy}$ denote a shortest path between $x$ and $y$ in $S$. For $u, v \in V(S)$, let $x \in A(S)$ be such that $w(u) \cdot d(x, u) = w(v) \cdot d(x, v)$ and for any two shortest paths $\pi_{x u}$ and $\pi_{x v}$ in $S$, $\pi_{x u} \cap \pi_{x v} = \varnothing$. For such a triple $(u, v, x)$, we define $\ct(u, v, x)$ as the cost of covering $u$ and $v$ by a center placed at $x$ i.e. $\ct(u, v, x) = w(u) \cdot d(x, u) = w(v) \cdot d(x, v)$. We define $R_S$ to be the set of all possible values of $\ct(u, v, x)$.

Let $X^* \subset A(S)$ be a set of $k$ optimal centers of the cactus stem $S$ with respect to the weight $k$-center problem and let $\lm^*$ be its optimal cost. We have the following lemma. 

\begin{lemma} \label{lem:cost}
    There exists a triple $(u, v, x)$ such that $\lm^* = \ct(u, v, x)$.
\end{lemma}

\begin{proof}
    Since $\lm^*$ is the optimal cost, their exists a center $x \in X^*$ that determines $\lm^*$. Let $u \in V(S)$ be such that $w(u) \cdot d(x, u) = \lm^*$ and the closest center to $u$ in $X^*$ is $x$. Since $x$ determines $\lm^*$, $u$ always exists. But that alone is not enough to determine $\lm^*$ since we can move $x$ towards $u$ and reduce the weighted distance of $x$ to $u$. Therefore, there exists another vertex $v \in V(S)$, such that $w(v) \cdot d(x, v) = \lm^*$. Also, no path $\pi_{u x}$ and $\pi_{v x}$ overlap, since again $x$ can be moved towards a particular direction which reduces the weighted distance of $x$ to both $u$ and $v$.  \qed
\end{proof}

We call the set $R_S$ the set of \emph{candidate values} of $\lm^*$. Based on $u, v$ and $x$ the candidate values in $R_S$ can be divided into the following types.

\begin{itemize}
    \item[] \tb{type 1:} $x$ does not lie inside any cycle of $S$,
    \item[] \tb{type 2:} $x$ lies on the minor hemisphere of some cycle $C_i$ in $S$ and \begin{itemize}
        \item[] \tb{(a):} $u$ and $v$ both belong to $V(\C_i)$,
        \item[] \tb{(b):} either $u$ or $v$ belongs to $V(
            C_i)$, but not both,
        \item[] \tb{(c):} neither $u$ nor $v$ belong to $V(\C_i)$.
    \end{itemize}
    \item[] \tb{type 3:} $x$ lies on the major hemisphere of some cycle $C_j$ in $S$ and \begin{itemize}
        \item[] \tb{(a):} $u$ and $v$ both belong to $V(\C_j)$,
        \item[] \tb{(b):} either $u$ or $v$ belongs to $V(
            \C_j)$, but not both,
        \item[] \tb{(c):} neither $u$ nor $v$ belong to $V(\C_j)$.
    \end{itemize}
\end{itemize}

In addition to $R_S$, we define another set of costs $R_S'$ as the set $\{ \lct(u, v) \mid \forall (u, v) \in V(S) \times V(S) \}$.
In the following sections we present our algorithm to compute the weighted $k$-center of $S$. The algorithm works in two stages. In Stage 1, we look at all types of candidate values in $R_S$ except type 3(c) as defined above. In Stage 2, we check for all candidate values in $R_S$ of type 3(c).

\subsection{Stage 1} \label{sec:stage1}

For a cycle $C_i$ of $S$, let $A_i$ represent its minor hemisphere and $B_i$ its major hemisphere. For a given cactus stem $S$, let $S'$ be the graph left after removing all pendant vertices of $S$. See Figure \ref{fig:stem-skeleton}.

\begin{figure}[ht]
    \centering
    \includegraphics[width = .8 \textwidth]{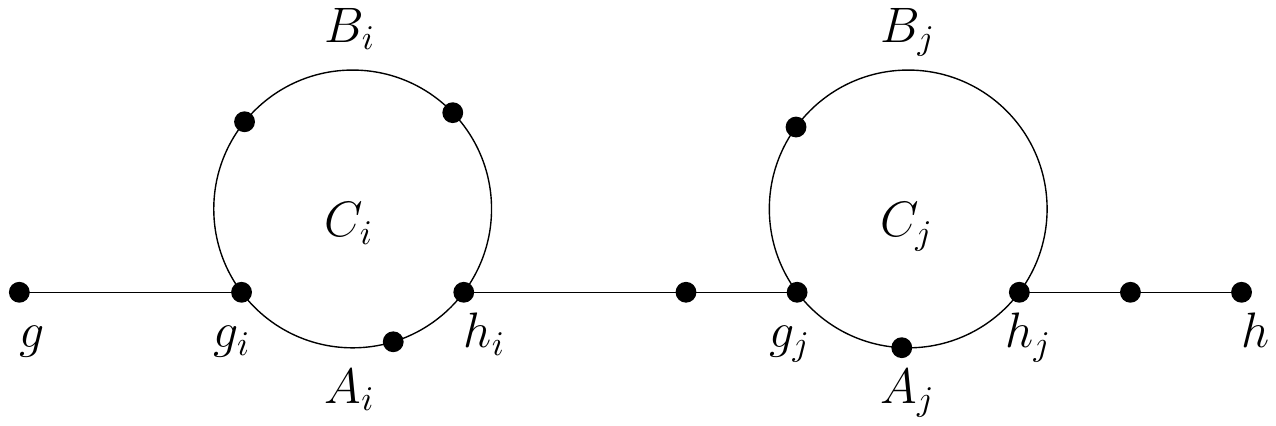}
    \caption{The graph $S'$ assuming $S$ to be the graph represented in Figure \ref{fig:cactusstem}.} \label{fig:stem-skeleton}
\end{figure}

We convert $S'$ into a tree $T_{S'}$ by removing $B_i$ from each cycle $C_i$ in $S'$ and attaching a path to $g_i$ and $h_i$ as shown in Figure \ref{fig:2-tree}. Note that Figure \ref{fig:2-tree} only shows the overall structure of $T_{S'}$ and some vertices are not depicted. The tree $T_{S}$ is generated by attaching all pendant vertices of $S$ into their respective locations in $T_{S'}$. See Figure \ref{fig:2-tree-pendant}.

\begin{figure}[ht]
    \centering
    \includegraphics[width = .9 \textwidth]{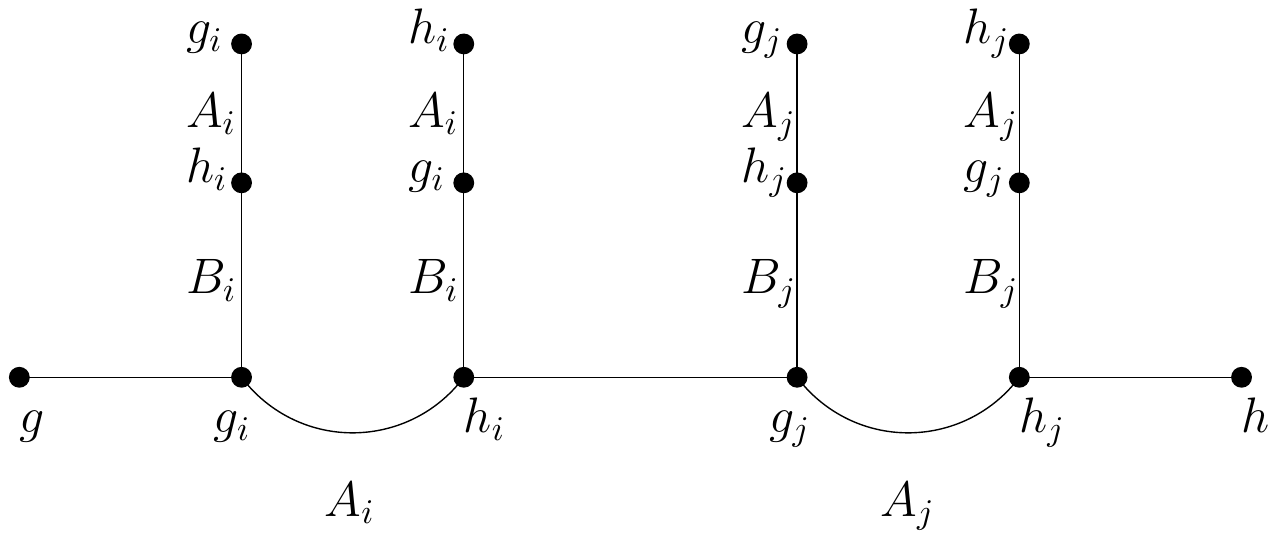}
    \caption{The structure of $T_{S'}$.} \label{fig:2-tree}
\end{figure}

\begin{figure}[ht]
    \centering
    \includegraphics[width = \textwidth]{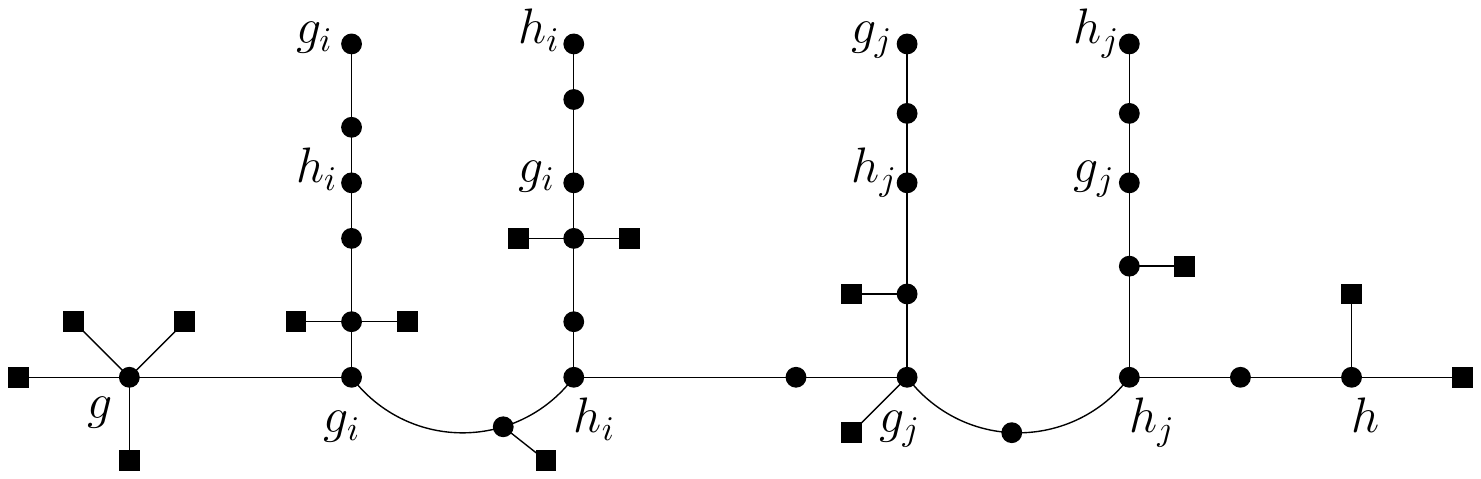}
    \caption{The tree $T_S$.} \label{fig:2-tree-pendant}
\end{figure}

Based on the vertices in $T_{S}$ we generate points on $x$-axis of an $xy$-coordinate system as follows. Fix two points $p_0$ and $q_0$ with respect to $g_0$ and $h_0$ arbitrarily, such that $p_0$ is to the left of $q_0$ and the distance between them is $d_{T_{S}}(g_0, g_0)$. For each vertex $u_i \in V(T_{S'})$, we place a point $p_i$ to the right of $p_0$ such that the distance between them is $d_{T_{S}}(g_0, u_i)$, and a point $q_i$ to the left of $q_0$ such that the distance between them is $d_{T_{S}}(u_i, h_0)$. For each vertex $v_j \in V(T_S) \setminus V(T_{S'})$, attached to a vertex $u_{j'} \in V(T_{S'})$, we place four points: one each to the left of $p_{j'}$ and $q_{j'}$ and one each to their right, at a distance $d_{T_S}(u_j, v_{j'})$. See Figure \ref{fig:2-tree-path} for reference.
We denote by $P_{S}$ the path represented by the left to right ordering of the points placed.

\begin{figure}[ht]
    \centering
    \includegraphics[width = .9 \textwidth]{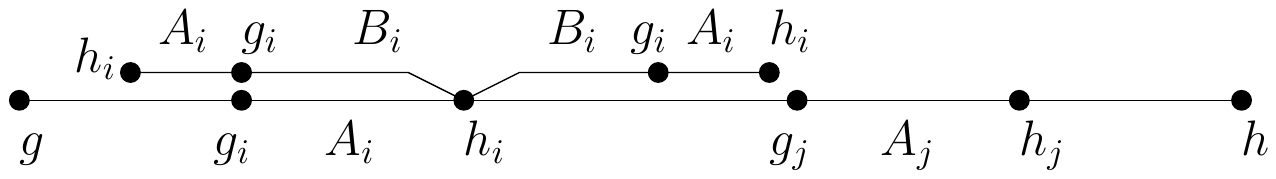}
    \caption{Some of the vertices of $S'$ mapped to a real line $\R$.} \label{fig:2-tree-path}
\end{figure}

Since, a single vertex of $S$ is represented by a constant number of nodes in $P_S$ we have $|V(P_S)| = O(n)$. Recall that $R_0(P_S) = \{\lct(u, v) \mid \forall (u, v) \in V(P_S) \times V(P_S)\}$ and $R_1(P_S) = \{\ct(u, v) \mid \forall (u, v) \in V(P_S) \times V(P_S)\}$. Notice that $R_0(P_S) = R_S'$ and $R_1(P_S)$ contains all types of candidate values of $R_S$ except type 3(c). We can use Lemma \ref{lem:paths} and the feasibility test of Section \ref{sec:feasibility} to find $\lm_1$ and $\lm_2$ such that $\lm_1 < \lm^* \leq \lm_2$ and no value in $R_0(P_S) \cup R_1(P_S)$ lies inside the range $(\lm_1, \lm_2)$. This step takes $O(n \log n)$ time since in our case both $m$ and $\tau$ equals $O(n)$. We can therefore make the following claim.

\begin{theorem}
    $\lm^* = \lm_2$ is the optimal cost of the $k$-center problem in $S$ if $\lm^*$ is not a type 3(c) candidate value of $R_S$. This cost can be computed in $O(n \log n)$ time, where $n$ is the size of the $S$.
\end{theorem}

\begin{definition}
    The \emph{$\lm$-cover} of a vertex $v \in V(S)$, denoted by $\Phi_{\lm} (v)$, is defined as
    \begin{equation*}
        \Phi_{\lm} (v) = \{y \in A(S) \mid w(v) \times d(y, v) \leq \lm\}.
    \end{equation*}
    The subscript $\lm$ from $\Phi_{\lm}(\cdot)$ is usually omitted if $\lm = \lm^*$.
\end{definition}

\begin{figure}[ht]
    \centering
    \includegraphics[width = .75 \textwidth]{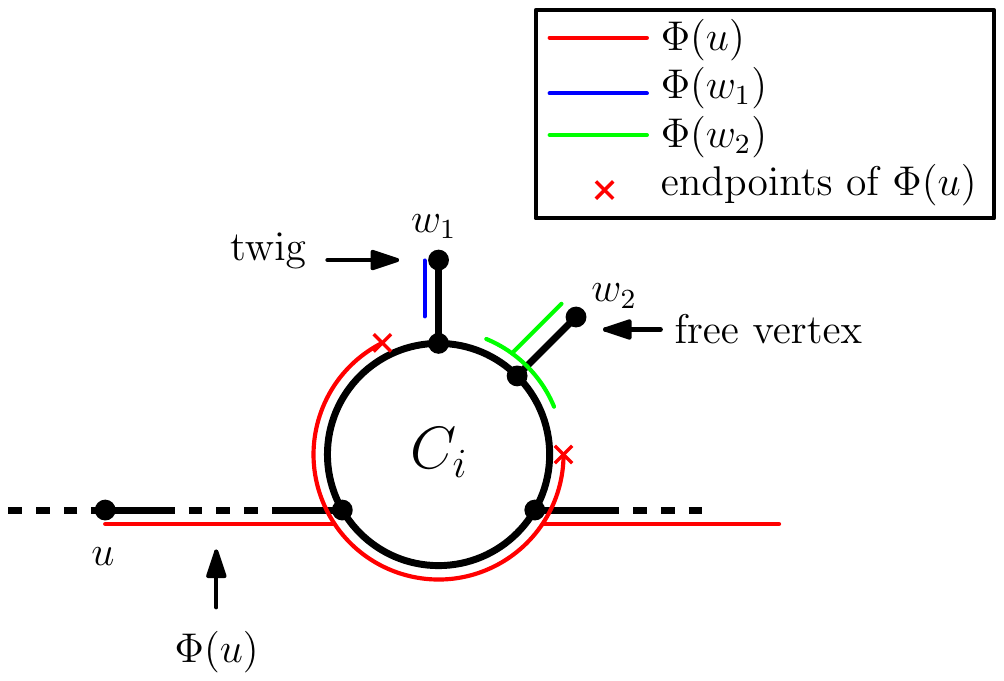}
    \caption{Visualization of various definition.} \label{fig:lambda_cover}
\end{figure}

The region $\Phi_{\lm} (v)$ can also be defined as all possible location in $S$ where a center can be placed with weighted distance to $v$ at most $\lm$.

\begin{definition}
    A point $y \in A(S)$ is an \emph{endpoint} of $\Phi_{\lm} (v)$ if $w(v) \times d(y, v) = \lm$. Again, the subscript $\lm$ is usually omitted if $\lm = \lm^*$.
\end{definition}

The endpoints of $\Phi_{\lm} (v)$ are all possible location in $S$ where a center can be placed with weighted distance to $v$ equal to $\lm$. See Figure \ref{fig:lambda_cover} for reference. We have the following results.

\todo{Check these two lemmas.}\begin{lemma} \label{lem:endpoint}
    For $u \in V(S)$ and $e \in E(S)$, $\Phi(u)$ has exactly one endpoint on $e$ if and only if $\Phi_{\lm}(u)$ has exactly one endpoint on $e$, for any $\lm \in (\lm_1, \lm_2)$.
\end{lemma}    

\begin{proof}
    Let an endpoint of $\Phi(u)$ lie on the edge $e = (u_e, v_e) \in E(S), u_e, v_e \in V(S)$. Note that the values  $\lct(u, u_e)$ and $\lct(u, v_e)$ both belong to $R_S'$ and therefore, do not lie in the range $(\lm_1,\lm_2)$. Let us assume, without loss of generality, that $\lct(u, u_e) \leq \lct(u, v_e)$.

    Since, $\Phi(u)$ has exactly one endpoint on $e$, we have $\lct(u, u_e) < \lm^* < \lct(u, v_e)$. This implies that for any $\lm \in (\lm_1, \lm_2)$, $\lct(u, u_e) < \lm < \lct(u, v_e)$ and therefore $\Phi_{\lm}(u)$ also has exactly one endpoint on $e$. The other direction can be proved similarly.  \qed
\end{proof}

\begin{lemma}
    For any two vertices $u, v \in V(S)$, let $e \in E(S)$ be an edge such that it does not belong to the major hemisphere of any cycle of $S$ in between $u$ and $v$. The regions $\Phi(u)$ and $\Phi(v)$ overlap on $e$ if and only if $\Phi_{\lm}(u)$ and $\Phi_{\lm}(v)$ overlap on $e$, for all $\lm \in (\lm_1, \lm_2)$.
\end{lemma}

\begin{proof}
    Let $e = (u_e, v_e)$, where $u_e, v_e \in V(S)$. Let $\al^*$ denote the overlapping region of  $\Phi(u)$ and $\Phi(v)$ on $e$. Similarly, let $\al$ denote the overlapping region of $\Phi_{\lm}(u)$ and $\Phi_{\lm}(v)$ on $e$. We need to show that $\al^*$ is nonempty if and only if $\al$ is nonempty.

    If $\al$ is empty then either $\Phi(u)$ or $\Phi(v)$ does not intersect $e$. Let us assume it is $\Phi(u)$. Then, $\lct(u, u_e), \lct(u, v_e) \leq \lm^*$.  Since, $\lct(u, u_e)$, $\lct(u, v_e) \notin (\lm_1, \lm_2)$ this implies $\lct(u, u_e)$, $\lct(u, v_e) \leq \lm$, for any $\lm \in (\lm_1, \lm_2)$. Therefore $\al$ is empty.
    
    If $\al^*$ is nonempty and contains $u_e$ then $\lct(u, u_e), \lct(v, u_e) \geq \lm^*$. Since, $\lct(u, u_e)$, $\lct(v, u_e) \notin (\lm_1, \lm_2)$ this implies $\lct(u, u_e)$, $\lct(v, u_e) \geq \lm$, for any $\lm \in (\lm_1, \lm_2)$. Therefore, $\al$ is also nonempty and contains $u_e$. The same can be proved with respect to the vertex $v_e$.

    If $\al^*$ is nonempty and does not contain either $u_e$ or $v_e$ then, there exists a point $y_e \in \al^*$ such that $w(u) \cdot d(u, y_e) = w(v) \cdot d(y_e, v)$ and for any two shortest paths $\pi_{y_e u}$ and $\pi_{y_e v}$ in $S$, $\pi_{y_e u} \cap \pi_{y_e v} = \varnothing$. This implies that $\ct(u, v, y_e) \in R(S)$. Since, $e$ does not belong to the major hemisphere of any cycle of $S$ in between $u$ and $v$, $\ct(u, v, y_e)$ is not a type 3(c) candidate value, as defined earlier, and so does not lie in the range $(\lm_1, \lm_2)$. Since $y_e \in \al^*$, $\ct(u, v, y_e) \geq \lm^*$, which in turn implies for any $\lm \in (\lm_1, \lm_2)$, $\ct(u, v, y_e) \geq \lm$. Therefore, $\al$ is also nonempty and contains $y_e$. \qed
\end{proof}

\subsection{Stage 2} \label{sec:stage2}

Now we show how to find $\lm^*$ if it is a type 3(c) candidate value of $R_S$. First we present a straightforward algorithm with running time $O(n^2 \log n)$. Later in Section \ref{sec:stage2_improve_1} and \ref{sec:stage2imp} we present two techniques by which we can  improve the running time of this algorithm further.

For each cycle $C_i$ of $S$ we construct a path $P_{C_i}$ as follows. Recall that $g_i$ and $h_i$ are the left and right hinge vertices of $C_i$, respectively. For each vertex $u_j$ to the left of $C_j$, we place a point $p_j$ at the $x$-axis of an $xy$-coordinate system with $x$-coordinate $-d(u_j, g_i)$. Similarly, for each vertex $v_{j'}$ to the right of $C_i$, we place a point $q_{j'}$ on the $x$-axis at the $x$-coordinate $\Delta_i + d(h_i, v_{j'})$, where $\Delta_i$ is the length of the major hemisphere of $C_i$. The left to right ordering of these points represent the path $P_{C_i}$.

Notice that the set $R' = \bigcup_i R(P_{C_i})$ contains all type 3(c) candidate values of $R_S$. Therefore we can use Lemma \ref{lem:paths} to find $\lm_1'$ and $\lm_2'$ such that $\lm_1' < \lm^* \leq \lm_2'$ and no value in $R'$ lies inside the range $(\lm_1', \lm_2')$. We update the previous $\lm_1$ and $\lm_2$ values to $\lm_1 \gets \max \{\lm_1, \lm_1'\}$ and $\lm_2 \gets \min \{\lm_2, \lm_2'\}$. Since now no value in $R_S$ lies inside the range $(\lm_1, \lm_2)$, we have $\lm^* = \lm_2$.

Let $\P$ represent the concatenation of the paths $P_{C_i}$, for all cycles $C_i$ in $S$. Since each path $P_{C_i}$ can be of size $O(n)$, the size of $\P$ can be $O(n^2)$, and the running time of our algorithm is $O(n^2 \log n)$.

\subsection{An Improvement to Stage 2} \label{sec:stage2_improve_1}

One idea to improve the running time of Stage 2 is to not search for $\lm^*$ over all possible type 3(c) candidate values but over the onces which are absolutely necessary. We achieve this by performing a feasibility test on $S$ and placing the centers optimally. Whenever we are unsure of the placement of a new center, we perform an additional feasibility test.

We perform the feasibility test on $S$ by first choosing a $\lm \in (\lm_1, \lm_2)$ (from Stage 1). The test processes the vertices of $S$ from left to right. At any intermediate step, the feasibility test may place centers on $S$. If the centers are placed on edges which do not belong to the major hemisphere of any cycle of $S$ then we continue to place those centers. This is because if any such center $x_i$ determines the optimal cost $\lm^*$, then that cost does not belong to the range $(\lm_1, \lm_2)$, since then $\lm^*$ is not a type 3(c) candidate value. Therefore, the vertices covered by $x_i$ will remain the same in the optimal case.

If the centers are placed on edges which do belong to the major hemisphere of some cycle $C_j$ of $S$, then we perform some extra steps. Let $x_l$ be such a center. There is a possibility that $x_l$ determines the optimal cost $\lm^*$. Let $u^*, v^* \in V(S)$ be two critical vertices of $x_l$ i.e. $\ct(u^*, v^*, x_l) = \lm^*$. The value $\ct(u^*, v^*, x_l)$ is a type 3(c) candidate value if neither $u^*$ nor $v^*$ belong to $V(\C_j)$. This implies that $\ct(u^*, v^*, x_l)$ may lie inside the range $(\lm_1, \lm_2)$, and we cannot provide the same guarantees as we did for the previous case. We therefore try to further reduce the range $(\lm_1, \lm_2)$ as follows.

We generate the path $P_{C_j}$ as mentioned earlier and use Lemma \ref{lem:paths} to find $\lm_1'$ and $\lm_2'$ such that $\lm_1' < \lm^* \leq \lm_2'$ and no value in $R(P_{C_j})$ lies inside the range $(\lm_1', \lm_2')$. We update the previous $\lm_1$ and $\lm_2$ values to $\lm_1 \gets \max \{\lm_1, \lm_1'\}$ and $\lm_2 \gets \min \{\lm_2, \lm_2'\}$.

This ensures that $\ct(u, v, x_l) \notin (\lm_1, \lm_2)$ for any $u, v \in V(S)$. We update our choice of $\lm$ to the new range $(\lm_1, \lm_2)$ and continue placing centers using our feasibility test algorithm. The center $x_l$ placed after this step will ensure that the vertices which covered by it will remain same in the optimal case. We continue the process until all $k$ centers have been placed. We make the following claims.

\begin{lemma}
    After completion of Stage 2, $\lm^* = \lm_2$.
\end{lemma}

\begin{proof}
    The only type 3(c) candidate values of $R(S)$ that we need to consider are $\ct(u, v, x)$, where $x$ lies on the major hemisphere of a cycle where our feasibility test wants to place a center. Since after completion of Stage 2 no such value lies inside the range $(\lm_1, \lm_2)$, therefore $\lm_2 = \lm^*$. \qed
\end{proof}

\begin{lemma}
    The time required by Stage 2 is at most $O(k n \log n)$.
\end{lemma}

\begin{proof}
    Whenever a center needs to be placed on the major hemisphere of some cycle $C_j$ of $S$ we generate the path $P_{C_j}$ and compute appropriate $\lm_1'$ and $\lm_2'$. Similar to Stage 1, this step takes $O(n \log n)$ time since $|V(P_{C_j})| = O(n)$. Since we are placing $k$ centers in $S$, such a step can repeat at most $k$ time. Therefore, the running time is $O(k n \log n)$. \qed
\end{proof}

\subsection{Further Improving Stage 2} \label{sec:stage2imp}

\subsubsection{Preliminaries}

Consider a feasibility test on $S$ with $\lm = \lm^*$ and root $h_0$, the rightmost vertex of $S$. Let $X^* \subset A(S)$ be the set of $k$ centers placed by the feasibility test. The set $X^*$ is an optimal solution to the weighted $k$-center problem in $S$.

Since we assume $\lm^*$ to be a type 3(c) candidate value of $R_S$, there exists $u^*, v^* \in V(S)$ and $x^* \in X^*$, such that $\lm^* = \ct(u^*, v^*, x^*)$. Also, $x^*$ lies on the major hemisphere of some cycle $C_{i^*}$ and $u^*, v^* \notin V(\C_{i^*})$. We consider the triplet $(u^*, v^*, x^*)$ for which $C_{i^*}$ is the leftmost possible cycle of $S$.

We assume $\Phi(u^*)$ and $\Phi(v^*)$ to completely cover the minor hemisphere of cycle $C_{i^*}$ as shown in Figure \ref{fig:vertex_cycle_map} and show some properties.

\begin{figure}[ht]
    \centering
    \includegraphics[width = .85 \textwidth]{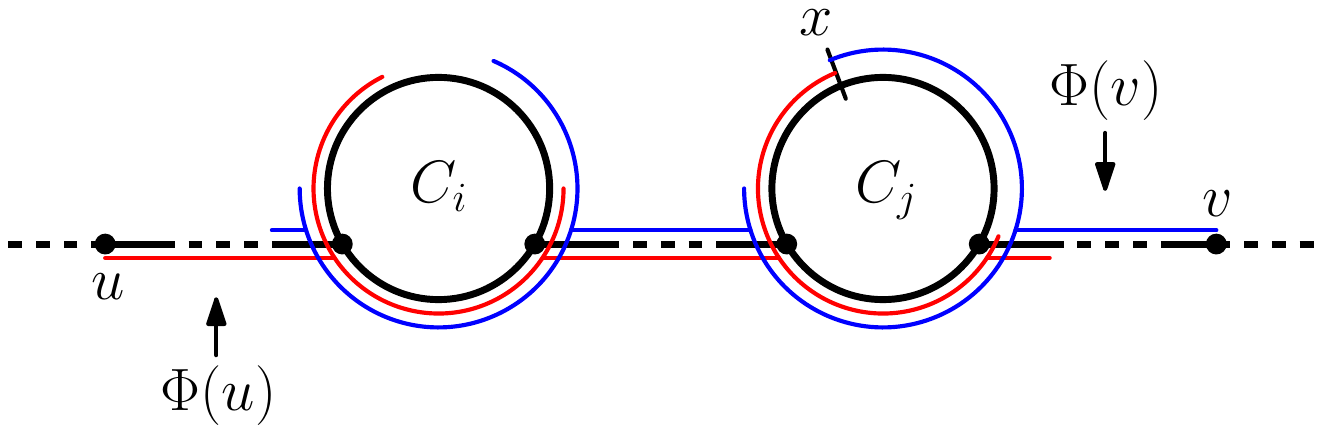}
    \caption{The $\lm$-regions of $u$ and $v$ where $\lm = \ct(u, v, x)$.} \label{fig:vertex_cycle_map}
\end{figure}

Arguing similar to the proof of Lemma \ref{lem:endpoint} we can show that $\Phi(v)$, for some $v \in V(S)$, completely covers the minor hemisphere of a cycle $C_i$ of $S$ if and only if $\Phi_{\lm_2} (v)$ also covers it completely.

\begin{definition}
    A \emph{twig} is any pendant (degree one) vertex $u$, adjacent to a vertex $v$, for which $v \notin \Phi(u)$, $u, v \in V(S)$.
\end{definition}

A twig is any pendant vertex $u$ of $S$ for which a center $x$ placed at its adjacent vertex has a weighted distance of more than $\lm^*$ from $u$. Arguing similar to the proof of Lemma \ref{lem:endpoint} we can also show that $u$ is a twig if $v \notin \Phi_{\lm_2} (u)$.

\begin{definition}
    A vertex $u$ is a \emph{free} vertex if $\Phi(u) \cap \Phi(v) = \varnothing$, for all twigs $v$ in $V(S)$.
\end{definition}

A free vertex is any vertex of $S$ which cannot be simultaneously covered along with a twig by a single center. For reference, see Figure \ref{fig:lambda_cover}.

\begin{definition}
    The \emph{$\lm$-intersection} of a cycle $C_i$, denoted by $\Theta_{\lm}(C_i)$, is defined as
    \begin{equation*}
        \Theta_{\lm}(C_i) = \bigcap_{v \in V'_i} \Phi_{\lm} (v),
    \end{equation*}
    where $V_i'$ represent the set of all free vertices in $V(\C_i)$.
Again, the subscript $\lm$ from $\Theta_{\lm}(\cdot)$ is omitted if $\lm = \lm^*$.
\end{definition}

\begin{figure}[ht]
    \centering
    \includegraphics[width = .75 \textwidth]{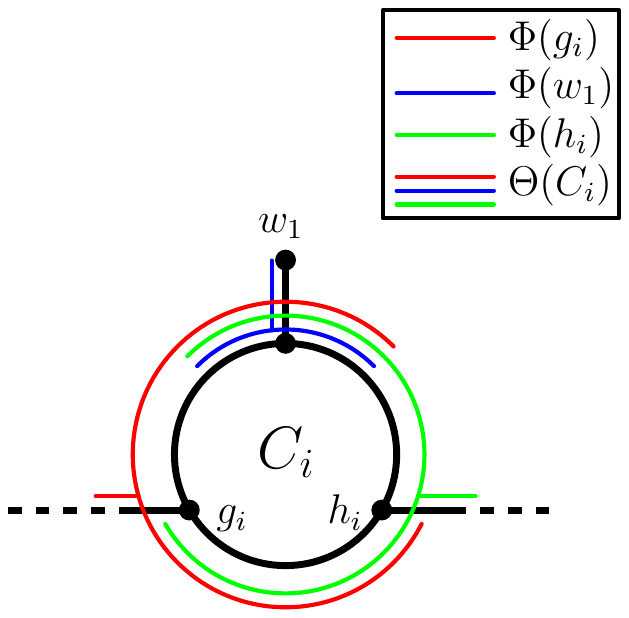}
    \caption{$\lm^*$-intersection of $C_i$.} \label{fig:theta}
\end{figure}

See Figure \ref{fig:theta} for reference. We are now ready to make the following claims.

\begin{lemma} \label{lem:intersection_region}
    For vertex $u^*$ and cycle $C_{i^*}$ the following two conditions hold: \emph{(a)} $\Phi(u^*) \cap \Theta(C_{i^*}) \neq \varnothing$ and \emph{(b)} $h_{i^*} \notin \Theta(C_{i^*})$.
\end{lemma}

\begin{proof}
    First we show that all vertices in $V_{i^*}'$ are covered by the center $x^*$. Let us assume that there is one more center $x'$ covering some vertices in $V_{i^*}'$ which are uncovered by $x^*$. Since all vertices in $V_{i^*}'$ are non-twig vertices, $x'$ can be placed on $C_{i^*}$ (and not inside any of its incident edges). But $\Phi(u) \cup \Phi(v)$ completely overlaps with all of $C_{i^*}$ and as a result $x'$ will also cover either $u^*$ or $v^*$. This implies that $\ct(u^*, v^*, x^*)$ will not determine $\lm^*$, which is a contradiction.

    Since all vertices in $V_{i^*}'$ and the vertex $u^*$ are covered by $x^*$, this implies that $\Theta(C_{i^*})$ is non-empty and intersects with $\Phi(u^*)$ which proves our first condition. If  $\Theta(C_{i^*})$ contains $h_{i^*}$ then the feasibility test would have placed $x^*$ on $h_{i^*}$ instead of on the major hemisphere of $C_{i^*}$. This would lead to $\ct(u^*, v^*, x^*)$ not being a type 3(c) candidate value, another contradiction. This proves our second condition. \qed
\end{proof}

\begin{lemma} \label{lem:cycle}
    No cycle $C_j$ between $u^*$ and $C_{i^*}$ satisfies both conditions of Lemma \ref{lem:intersection_region} simultaneously.
\end{lemma}

\begin{proof}
    Let us assume there exists a cycle $C_j$, in between vertex $u^*$ and cycle $C_{i^*}$ which satisfies both conditions of Lemma~\ref{lem:intersection_region} simultaneously. Since $u^*$ is covered by $x^*$ which is to the right of $C_j$, all vertices in $V_j'$ are either covered by $x^*$ or by at most one center placed on $C_j$. This condition holds because $\Theta(C_j)$ is non-empty and $V_j$ has no twigs.

    If all vertices in $V_j'$ are covered by $x^*$ then $h_j \in \Theta(C_j)$, which is a contradiction. If $h_j \notin \Theta(C_j)$, then one center say $x'$ needs to be placed on $C_j$ to cover the vertices in $V_j'$. But since $\Phi(u^*)$ intersects $\Theta(C_j)$ the feasibility test, while processing $C_j$, will place $x'$ in such a way that it also covers $u^*$. This is a contradiction since $\ct(u^*, v^*, x^*)$ will not determine $\lm^*$ then. \qed
\end{proof}

From Lemma \ref{lem:intersection_region} and \ref{lem:cycle} we can make the following observation.

\begin{observation}
    Let $u, v \in V(S)$ and $x$ belongs to the major hemisphere of a cycle $C_i$ such that $\Phi(u)$ completely overlaps with the minor hemisphere of $C_i$. We consider only those type 3(c) candidate values $\ct(u, v, x)$ where $C_i$ is the first cycle to the right of $u$ which satisfies both conditions in Lemma \ref{lem:intersection_region}.
\end{observation}

Consider an $xy$-coordinate system $\R^2$. For each vertex $v_i \in V(S)$, we place a point $p_i$ on the $x$-axis of $\R^2$ at the $x$-coordinate $d(v_i, h_0)$. For each such point $p_i$ we place a line $l_i$ on it with slope $w(v_i)$. Let $L_S$ be the set of these lines. We have the following definition.

\begin{definition}
    The \emph{rank} of a vertex $v_i \in V(S)$ is the respective position of line $l_i$ in a left to right ordering of the intersections of all lines in $L_S$ with the horizontal line $y = \lm^*$.
\end{definition}

Note that the ranks of the vertices do not change if we replace $\lm^*$ with $\lm_2$.

Let $u, v \in V(S)$ be such that $\rank(u) \leq \rank(v)$. Let $C_i$ be a cycle to the right of both $u$ and $v$. From the definition of rank we have the following observation.

\begin{observation} \label{obs:rank}
    $\Phi(u) \cap \A(C_i) \subseteq \Phi(v) \cap \A(C_i)$.
\end{observation}

\subsubsection{Finding the first cycle.}

We present an algorithm to find for each vertex $v \in V(C)$ the first cycle to its right which satisfies Lemma \ref{lem:intersection_region}, if it exists. Let $\mathcal C$ be the set of cycles in $S$ ordered from left to right. We compute $\Theta_{\lm_2}(C_i)$ for all cycle $C_i$ in $\mathcal C$. We remove from $\mathcal C$ any cycle $C_j$ for which either $\Theta_{\lm_2}(C_j) = \varnothing$ or $h_j \in \Theta_{\lm_2}(C_j)$. This is because for any point $x \in A(C_j)$ a type 3(c) candidate value $\ct(u, v, x)$ will never satisfy Lemma \ref{lem:intersection_region}.

We process the vertices in $V(S)$ with increasing rank and perform the following step. Let $v_i$ be the vertex we are currently processing. Find the leftmost cycle $C_j \in \mathcal C$ which is to the right of $u_i$. If no cycle exists we skip $u_i$ and continue to the next higher ranked vertex. This implies that no cycle exists to the right of $u_i$ which satisfies Lemma \ref{lem:intersection_region}.

If cycle $C_j$ exists, then we check if $\Phi_{\lm_2}(u_i) \cap \Theta_{\lm_2}(C_j) \neq \varnothing$. If the condition holds, we correctly identify $C_j$ as the first vertex to the right of $u_i$ which satisfies Lemma \ref{lem:intersection_region}. We continue to the next higher ranked vertex.

Otherwise, if the condition is not satisfied, we remove $C_j$ from $\mathcal C$. This is because from Observation \ref{obs:rank}, for any vertex $v_j$ with $\rank(v_j) \geq \rank(v_i)$,  $\Phi_{\lm_2}(u_i) \cap \Theta_{\lm_2}(C_j) = \varnothing$. We continue to the next cycle in $\mathcal C$ to the right of $C_j$. We stop when either $\mathcal C$ is empty or when all vertices in $V(S)$ have been processed.

When a vertex is successfully mapped to a cycle we charge the operation on the vertex and when a cycle is unsuccessfully compared for mapping with a vertex we charge the operation on the cycle. Since, each vertex and cycle is charged exactly once the running time of this procedure is $O(n \log n)$. The $\log n$ factor comes in since we have to identify the first cycle in $\mathcal C$ to the right of a particular vertex which takes $O(\log n)$ time.

\subsubsection{The improved algorithm.}

Here is the final algorithm which improves upon the running time of Stage 2 considerably. Consider an $xy$-coordinate system $\R^2$. For each vertex $v_i \in V(S)$, we place a point $p_i$ on the $x$-axis of $\R^2$ at a distance of $d(v_i, h_0)$ to the left of the origin. Again, for each vertex $v_i \in V(S)$, let $C_{j_1}$ be the first cycle to the right of $v_i$ which satisfies Lemma \ref{lem:intersection_region} and let $C_{j_2}$ be the cycle to the right of $v_i$ where $\Phi(v_i)$ has an endpoint on its minor hemisphere. We place points $p_{j_1'}$ and $p_{j_2'}$ on the $x$-axis at a distance of $d(v_i, g_{j_1}) + \Delta_{j_1} + d(h_{j_1}, h_0)$  and $d(v_i, g_{j_2}) + \Delta_{j_2} + d(h_{j_2}, h_0)$ to the left of the origin, respectively, where $\Delta_i$ is the length of the major hemisphere of cycle $C_i$.

Let $P_S'$ be the path formed by a left to right ordering of these points. We can make the following observation.

\begin{observation}
    Every relevant type 3(c) candidate value of $R_S$ is included in $R(P_S')$.
\end{observation}

Therefore, using Lemma \ref{lem:paths}, we can find $\lm_1'$ and $\lm_2'$ such that (a) $\lm_1' < \lm^* \leq \lm_2$, and (b) no relevant type 3(c) candidate value is inside the range $(\lm_1', \lm_2')$. Finally we update $\lm_1\gets \max \{\lm_1, \lm_1'\}$ and $\lm_2 \gets \min \{\lm_2, \lm_2'\}$. Since we have exhausted all candidate values of $\lm^*$, $\lm^* = \lm_2$. Since the size of $P_S'$ is $O(n)$ we have the following theorem.

\begin{theorem}
    The weighted $k$-center problem on a cactus stem of size $n$ can be solved in $O(n \log n)$ time.
\end{theorem}

\section{The Main Algorithm} \label{sec:main}

We are now ready to present our algorithm for the $k$-center problem on the cactus $G$. Initialize $\lm_1 \gets 0$ and $\lm_2 \gets \infty$. Let $G'$ be the graph formed by removing all pendant vertices from $G$ and let $T$ be its tree representation. Note that $t_r$ (a leaf node) and its associate vertex $v_r$ represent the roots of $T$ and $G$, respectively.

\begin{definition}
    A \emph{maxpath} of a rooted tree is any subpath in the tree whose end nodes are either of degree 1, or 3, or more, and whose internal nodes are of degree 2. A \emph{leaf maxpath} is any maxpath of which contains a leaf of the tree (other than the root).
\end{definition}

For each leaf maxpath $\pi_i$ of $T$, we do the following. Let $t_1, t_2, \ldots, t_m$ be the nodes in $\pi_i$ such that $t_1$ is a leaf. If $t_m$ represents a cycle of $G$, then we set $\pi_i'$ to path $\pi_i$ with node $t_m$ removed; otherwise $\pi_i'$ is same as $\pi_i$. Let $S'_i$ be the subgraph of $G$ that $\pi_i'$ represents and let $S_i$ be the graph formed by attaching all pendant vertices of $G$ adjacent to $S'_i$ to it. If $\pi_i'$ contains the node $t_m$, then $t_m$ is either a graft or a hinge. Let $v_m \in V(G)$ be the associate vertex of $t_m$. We do not include the pendant vertices adjacent to $v_m$ in $S_i$. We call $S_i$ a \emph{leaf cactus stem} of $G$. See Figures \ref{fig:leaf_path} and \ref{fig:leaf_stem} for reference.

\begin{figure}[ht]
    \centering
    \includegraphics[width = .8 \textwidth]{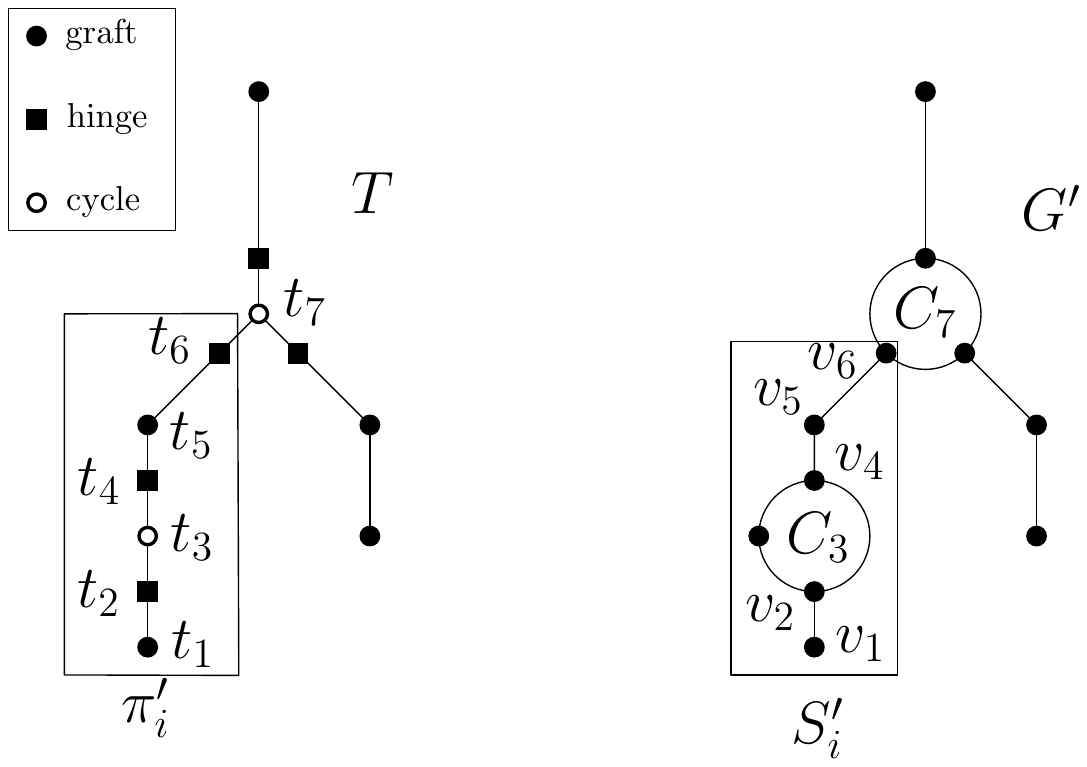}
    \caption{The nodes from $t_1$ to $t_7$ represent a leaf maxpath $\pi_i$ of $T$. The subpath $\pi_i'$ is from $t_1$ to $t_6$, as $t_7$ represents $C_7$. The subgraph $S_i'$ is from $v_1$ to $v_6$.} 
    \label{fig:leaf_path}
\end{figure}

\begin{figure}[ht]
    \centering
    \includegraphics[width = .33 \textwidth]{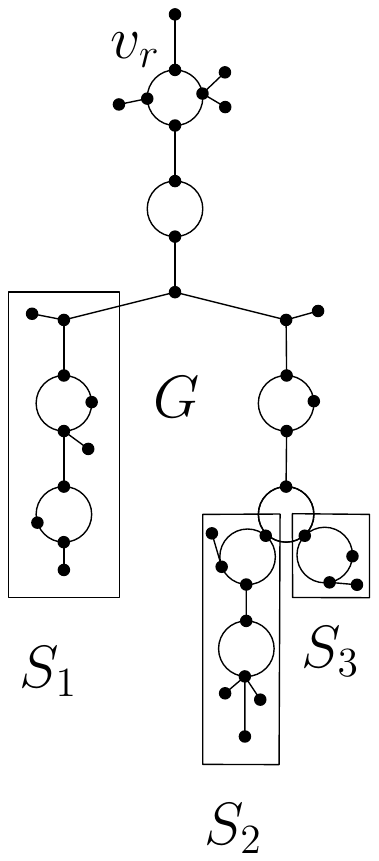}
    \caption{The leaf cactus stems $S_1, S_2$ and $S_3$ of $G$ with root $v_r$.} \label{fig:leaf_stem}
\end{figure}

For all leaf cactus stems $S_i$, we perform Stage 1 and Stage 2 of Section \ref{sec:stem}. For Stage 1, we generate $P_{S_i}$ for each $S_i$. We use Lemma \ref{lem:paths} to find the corresponding $\lm_1'$ and $\lm_2'$ values. For stage 2, we can process each $S_i$ sequentially to find the updated $\lm_1'$ and $\lm_2'$ values. It can be showed that $\lm_1' < \lm^* \leq \lm_2'$ and no value in $\bigcup_i R_{S_i}$ lies inside the range $(\lm_1', \lm_2')$. We update $\lm_1$ and $\lm_2$ to $\lm_1 \gets \max \{\lm_1, \lm_1'\}$ and $\lm_2 \gets \min \{\lm_2, \lm_2'\}$.

We now do a feasibility test on $G$ by choosing a $\lm \in (\lm_1, \lm_2)$. Let $X$ be the set of centers placed by the feasibility test. This defines a partition on $G$ where each vertex is associated with its nearest center. By our definition of $\lm_1$ and $\lm_2$, for each cactus stem $S_i$ in $G$, the partition defined by $X$ on $S_i$ is the same as that defined by an optimal set of centers $X^*$. This property was extensively exploited by Frederickson et al. \cite{frederickson} in the design of their optimal algorithm. We process each leaf cactus stem $S_i$ of $G$ and convert it into a pendant vertex of $G$ as follows.

Let $h \in V(S_i)$ be the vertex through which $S_i$ is connected to the rest of the graph $G$. Let $x_{h} \in X$ be the center closest to $h$ and let $V' \subseteq V(S)$ be the set of vertices for which $x_{h}$ is the nearest center. We replace $S$ by a pendant vertex $h'$ attached to $h$ such that $h' = \arg \min_{v} \{\lm / w(v) - d(h, v) \mid \forall v \in V'\}$ and set the length of the edge $(h, h')$ to $d(h, h')$. Let $X_{S_i} \subseteq X$ be the set of centers placed in $S_i$. Since $\lm \in (\lm_1, \lm_2)$, the vertices in $V(S_i) \setminus V'$ will also be covered by $X_{S_i} \setminus \{x_{h}\}$ in the optimal case. Therefore, we can be certain that these centers should also be included in the optimal set.

We decrement from $k$ the number of centers thus placed i.e. $|X_{S_i} \setminus \{x_{h}\}|$. We continue the process recursively in the new cactus until a single cactus stem remains. We can compute the weighted $k$-center in this last stem by using the algorithm in Section \ref{sec:stem}.  We have the following result.

\begin{theorem}
    The weighted $k$-center problem on a cactus $G$ can be solved in $O(n \log^2 n)$ time, where $n$ is the number of vertices of $G$.
\end{theorem}

\begin{proof}
    In each iteration, performing the steps of Stage 1 and Stage 2 takes $O(n \log n)$ time. Moreover, in every iteration, all leaf cactus stems of the graph get reduced to a pendant vertex. This may result in new cactus stems becoming leaf cactus stems. However, the number of such new cactus stems can be at most half of the number of previous leaf cactus stems. This results in at most $O(\log n)$ many iterations. Therefore, the total time taken by our algorithm over all iterations is $O(n \log^2 n)$. \qed
\end{proof}

\section{Conclusion} \label{sec:conclusion}

In this paper we presented an $O(n \log^2 n)$ time algorithm to solve the weighted $k$-center problem in a cactus of size $n$. We adapt and use Frederickson and Johnson's \cite{frederickson1983} parametric search technique, generalized by Wang and Zhang \cite{wang}, to achieve our result. In their papers they had also shown how to iteratively optimize the feasibility test algorithm to achieve a sublinear running time. It is worth pursuing to adapt this idea to find a faster feasibility test algorithm for our problem. This has the potential of further reducing the running time of our algorithm by a factor of $\log n$ which we believe will be optimal.

\bibliographystyle{plain}
\bibliography{cactus}

\end{document}